\documentclass[letterpaper,11pt]{article}
\pdfoutput=1
\usepackage{bm} %bold math
\usepackage[table,xcdraw,svgnames, dvipsnames]{xcolor}
\usepackage{microtype}
\usepackage[utf8]{inputenc}
\usepackage{graphicx,fullpage,paralist}
\usepackage{amsmath, amssymb, amsthm}
\usepackage{comment,hyperref}
\usepackage{thmtools}
\usepackage{tikz}
\usepackage{pgfplots}
\usepackage{varwidth}
\pgfplotsset{compat=1.10}
\usepgfplotslibrary{fillbetween}
\usetikzlibrary{backgrounds}
\usetikzlibrary{patterns}
\usetikzlibrary{positioning}
\usepackage{nicefrac}
\usepackage{enumitem}
\usepackage{geometry}
\usepackage{array}
\usepackage{multirow}
\usepackage{stmaryrd}
\usepackage[font=small]{caption}
\usepackage{makecell}
\usepackage{bbold}
\usepackage{cite}
\usepackage{booktabs}
\usepackage{mathtools}
\usepackage{thm-restate}
\usepackage{titling}
\usepackage{cleveref}
\usepackage[rightcaption]{sidecap} % Captions appear nex to figure.
\usepackage[skip=2pt, indent=20pt]{parskip}

\newcommand{\calM}{\mathcal{M}}

\newcommand{\cL}{\mathcal{L}}

\newcommand{\calE}{\mathcal{E}}
\newcommand{\calI}{\mathcal{I}}

\newcommand{\RR}{\mathbb{R}}
\newcommand{\NN}{\mathbb{N}}
\newcommand{\N}{\mathbb{N}}

\newcommand{\Bin}{\textrm{Bin}}
\newcommand*\diff{\mathop{}\!\mathrm{d}}

\newcommand{\Gen}{\normalfont\textrm{Gen}}

\DeclareMathOperator{\AUX}{\normalfont{AUX}}
\DeclareMathOperator{\OPT}{\normalfont{OPT}}
\DeclareMathOperator{\GRE}{\normalfont{ALG_0}}
\DeclareMathOperator{\PAR}{\normalfont{ALG_1}}
\DeclareMathOperator{\ALG}{\normalfont{ALG}}

\DeclareMathOperator*{\argmin}{arg\,min}

\DeclareMathOperator*{\Para}{Par}
\def\set#1{\left\{ #1 \right\}}
\def\abs#1{\left| #1 \right|}

\def\prn#1{\left( #1 \right)}
\def\brk#1{\left[ #1 \right]}
\def\midd{\:\middle|\:}
\def\eps{\varepsilon}
\def\cM{\mathcal{M}}
\def\cF{\mathcal{F}}

\def\cL{\mathcal{L}}
\def\cI{\mathcal{I}}
\def\cC{\mathcal{C}}

\def\argmin{\operatornamewithlimits{arg\,min}}

\theoremstyle{plain}
\newtheorem{theorem}{Theorem}
\newtheorem{lemma}{Lemma}

\newtheorem{proposition}[theorem]{Proposition}
\newtheorem{claim}{Claim}
{\bfseries}{\itshape} % For numbered claims.
\newtheorem*{objective*}{Objective}

\theoremstyle{definition}
\newtheorem{definition}{Definition}

\newtheorem{assumption}{Assumption}

\usepackage[textsize=small,textwidth=2cm]{todonotes}

\usepackage{multirow}
\usepackage[noend]{algpseudocode}
\usepackage{algorithmicx,algorithm}

\newlength{\algofontsize}
\hypersetup{
	colorlinks,
	linkcolor={red!50!black},
	citecolor={blue!50!black},
	urlcolor={blue!80!black}
}
\begin{document}
\algrenewcommand\algorithmicrequire{\textbf{Input:}}
\algrenewcommand\algorithmicensure{\textbf{Output:}}
	
\title{Matroid Secretary via Labeling Schemes
\vspace{.7cm}
}

\thanksmarkseries{alph}
\author{Krist\'{o}f B\'{e}rczi\thanks{MTA-ELTE Matroid Optimization Research Group and HUN-REN--ELTE Egerv\'{a}ry Research Group, \\Department of Operations Research, E\"{o}tv\"{o}s Lor\'{a}nd University, Budapest, Hungary {\tt (kristof.berczi@ttk.elte.hu)}}
\and Vasilis Livanos\thanks{Center for Mathematical Modeling, Universidad de Chile, Santiago, Chile. {\tt (vas.livanos@gmail.cl)}}
\and Jos\'{e} A.~Soto\thanks{Department of Mathematical Engineering and Center for Mathematical Modeling, Universidad de Chile,\\ Santiago, Chile {\tt (jsoto@dim.uchile.cl)}}
\and Victor Verdugo\thanks{Institute for Mathematical and Computational Engineering, and Department of Industrial and Systems \\Engineering, Pontificia Universidad Católica de Chile, Santiago, Chile {\tt (victor.verdugo@uc.cl)}}
}

\date{}

\maketitle
\begin{abstract}
The matroid secretary problem (MSP) is one of the most prominent settings for online resource allocation and optimal stopping. A decision-maker is presented with a ground set of elements $E$ revealed sequentially and in random order.
Upon arrival, an irrevocable decision is made in a take-it-or-leave-it fashion, subject to a feasibility constraint on the set of selected elements captured by a matroid defined over $E$. The decision-maker only has ordinal access to compare the elements, and the goal is to design an algorithm that selects every element of the optimal basis with probability at least $\alpha$ (i.e., $\alpha$-probability-competitive). While the existence of a constant probability-competitive algorithm for MSP remains a major open question, simple greedy policies are at the core of state-of-the-art algorithms for several matroid classes.

We introduce a flexible and general algorithmic framework to analyze greedy-like algorithms for MSP based on constructing a language associated with the matroid. Via this language, we establish a lower bound on the probability-competitiveness of the algorithm by studying a corresponding Poisson point process that governs the words' distribution in the language. Using our framework, we break the state-of-the-art guarantee for laminar matroids by settling the probability-competitiveness of the greedy-improving algorithm to be exactly $1-\ln(2)\approx 0.3068$. We also showcase the capabilities of our framework in graphic matroids, to show a probability-competitiveness of $0.2693$ for simple graphs and $0.2504$ for general graphs.
\end{abstract}

\newpage
\thispagestyle{empty}

\newpage
\setcounter{page}{1}

%\input{sections/introduction.tex}
%\input{sections/preliminaries.tex}
%\input{sections/labeling.tex}
%\input{sections/laminar.tex}
%\input{sections/graphic.tex}
%\input{sections/acknowledgement.tex}

%%%%%%%%%%%%%%%%
\section{Introduction}
\label{sec:intro}
%%%%%%%%%%%%%%%%

One of the most celebrated problems in online decision-making is the \emph{secretary problem}, which aims to hire the best out of $n$ candidates who arrive in a uniformly random order. The decision-maker can only rank the candidates observed up to that point and, upon observing a candidate, must make an immediate and irrevocable decision to accept or reject them. 
The famous optimal strategy of rejecting the first $1/e$ fraction of candidates and then selecting the first one who is better than all previously seen candidates guarantees at least a $1/e$ probability of hiring the best candidate (see, e.g., Lindley~\cite{lindley}, Dynkin~\cite{dynkin}).

Arguably, one of the most interesting generalizations is the \emph{matroid secretary problem} (MSP) introduced by Babaioff et al.~\cite{mat-sec}. In the MSP, the decision-maker is given a matroid and can hire any set of candidates, forming an independent set in the matroid. An algorithm that, for a fixed optimal basis, selects a set in which every candidate from the optimal basis appears with probability at least $\alpha$ is called \emph{$\alpha$-probability-competitive}. For the remainder of this paper, we assume, without loss of generality, that a unique optimal basis exists.
When the decision-maker observes a weight associated with each candidate and the selected set weights at least an $\alpha$ fraction of the optimal weight basis, the algorithm is \emph{$\alpha$-utility-competitive}. Over the last two decades, there has been a series of works studying the competitiveness for general matroids, resulting in a $O(\log{\log{r}})$ utility-competitive and $O(\log r)$ probability-competitive algorithms~\cite{mat-sec,sqrt-log-r-matroid-sec,loglogrank-matroid-sec1,loglogrank-matroid-sec2,forbidden-paper} for rank $r$ matroids. Babaioff et al.~\cite{mat-sec} conjectured the existence of a (utility) constant-competitive algorithm for every matroid. The \emph{matroid secretary conjecture}, as it has been known since, has been resolved for several sub-classes of matroids~\cite{forbidden-paper,dinitz-secretary,KesselheimRTV13,korula-pal-graphic,soto-secretary,ImW11}, but it remains open in general and even for fundamental sub-classes such as gammoids or binary matroids. 

%%%%%%%%%%%%%%%%
\subsection{Our Contributions and Techniques}
%%%%%%%%%%%%%%%%

In this work, we develop new tools for analyzing algorithms for the MSP that are more effective at managing the complex correlations involved in determining whether an element arriving at a specific time can be selected. Conceptually, our main contribution is introducing a new framework to analyze greedy-improving algorithms for the MSP using labeling schemes. Using this, we improve the state-of-the-art probability-competitiveness for the fundamental laminar and graphic matroid classes and match the best-known probability-competitiveness guarantees for several other matroids since, as we discuss in Appendix~\ref{app:qforbidden}, our labeling scheme framework extends the capabilities of the forbidden sets technique by Soto, Turkieltaub, and Verdugo~\cite{forbidden-paper}.

\paragraph{\bf The Labeling Scheme Framework.}The algorithms we analyze initially skip the first $p$ fraction of elements. Afterward, they mark every element that improves upon the current best set, i.e., an \emph{improving element}. Some marked elements will be added to the output $\ALG$ when they arrive, according to an internal property $P$ specific to the matroid. Property $P$ filters the improving elements and guarantees the feasibility of the resulting selection, i.e., if $P$ holds when an element $e$ arrives, then it must be that $\ALG + e$ is independent. To analyze the algorithm, we consider a separate procedure, called a \emph{labeling scheme} that visits the marked elements in reverse order and assigns a label to each. At the end of the procedure, we obtain a word $z$ by concatenating all labels assigned by our scheme. 

The core idea is that, for every matroid $\cM$, we can assign a language $\cL(\cM)$ such that, for any fixed element $e^*$ in the optimal set,
$\Pr[e^* \in \ALG] \geq \Pr[z \in \cL(\cM)]$,
where $\ALG$ is the final set of our algorithm. Then, the competitiveness guarantee follows by computing $\Pr[z \in \cL(\cM)]$, and we show that the uniformly random arrival order of the elements implies the characters (i.e., the labels) of $z$ are chosen uniformly at random from the label set and the length of $z$ follows a Poisson distribution.
In \Cref{sec:label} we introduce our framework and the stochastic analysis tools used throughout this work. 

\paragraph{\bf Laminar Matroids.} One of the first algorithms proposed for the MSP was the \emph{greedy-improving} algorithm, which greedily selects every improving element, subject to feasibility. Despite its simplicity, it has been remarkably successful for many matroid classes. For laminar matroids, the best-known guarantee before our work was the very recent $\approx 0.2105$-probability-competitiveness, achieved by Huang, Parsaeian, and Zhu~\cite{zahra-laminar-secretary} using the greedy-improving algorithm. Using our labeling scheme framework, we settle the probability-competitiveness of the greedy-improving algorithm for laminar matroids and show a tight guarantee of $1-\ln(2)\approx 0.3068$ (Theorem \ref{thm:laminar-intro}). In fact, for every natural number $r$, we characterize the optimal competitiveness attained by the greedy-improving algorithm on laminar matroids of rank $r$. The competitiveness decreases as $r \to \infty$, with a limit of $1 - \ln(2)$. We also obtain a full characterization of the greedy-improving algorithm on uniform matroids for any fixed rank $r$ and show that it grows from $1/e$ when $r=1$ to $1-1/e$ as $r\to \infty$ (Proposition \ref{lem:gi}).
In particular, we show that the probability-competitiveness of greedy-improving for $r=2$ outperforms the approximation achievable by the so-called {\it optimistic} algorithm introduced by Babaioff et al.~\cite{babaioff2007knapsack} and recently fully analyzed for the case of $r=2$ by Albers and Ladewig \cite{albers2021new}.

Towards an answer for the matroid secretary conjecture, we study the simplest matroids for which the existence of a $1/e$-probability-competitive algorithm is still open: rank-$2$ matroids. As they are laminar matroids, we can specialize our results to this class, obtaining a probability-competitiveness of 0.3341.  
We show one can go beyond this via an alternative mixed algorithm to get a $0.3462$-probability-competitive algorithm for the MSP on rank-2 matroids (\Cref{thm:rank-2-intro}). The analysis of laminar MSP via our labeling schemes framework can be found in \Cref{sec:laminar}.

\paragraph{\bf Graphic Matroids.} In the graphic case, we adjust our labeling scheme to account for an algorithm that is not greedy-improving but avoids selecting edges that seem like they will induce cycles in the final set. The previous best algorithm in \cite{forbidden-paper} is $1/4$-probability-competitive as it satisfies the $2$-forbidden property. 
We showcase the capabilities of our framework to show that one can also improve the 2-forbidden guarantee in this case.

We begin with a basic algorithm that orients the improving edges and greedily constructs an auxiliary digraph $\AUX$ with a maximum in-degree of one. Consequently, $\AUX$ has at most one cycle per connected component. The algorithm in \cite{forbidden-paper} selects a subset of $\AUX$ that includes only edges connecting nodes with in-degree 0 when they appear. Hence, the resulting set does not contain the last edge of any cycle and is always feasible. However, this procedure may drop too many arcs from $\AUX$. Our labeling schemes allow a more careful study of which edges we must delete to ensure that the graph is acyclic. For that, we assign a \emph{non-negative generation} to each arc in $\AUX$ according to the order in which arcs were added. The basic algorithm outputs only edges of generation 0. We show that, in fact, by keeping edges with generation different from 1, we also obtain an acyclic graph. This new \emph{generation algorithm} is $0.2693$-probability-competitive for simple graphs (i.e., without parallel edges) -- note that even for simple graphs, the previous algorithm by \cite{forbidden-paper} is still $1/4$-probability-competitive. By randomizing between the generation algorithm and an oblivious algorithm that is good for those optimal edges $e^*$ that are likely to belong to some 2-cycles in $\AUX$, we obtain our final algorithm for graphic matroids that achieves a $0.2504$-probability-competitiveness for general graphs (Theorem~\ref{thm:graphic-intro}). The analysis of graphic MSP via our labeling schemes framework can be found in \Cref{sec:graphic}.

\subsection{Related Work}
%%%%%%%%%%%%%%%%

The MSP was introduced by Babaioff et al.~\cite{mat-sec}, in which the authors gave a $\Omega({1/\log{r}})$-utility-competitive algorithm for general matroids. Later, Chakraborty and Lachish \cite{sqrt-log-r-matroid-sec} gave an improved $\Omega({1/\sqrt{\log{r}}})$-utility-competitive algorithm, before the current best bound of $\Omega({1/\log{\log{r}}})$ obtained first by Lachish \cite{loglogrank-matroid-sec1} and then by Feldman, Svensson, and Zenklusen \cite{loglogrank-matroid-sec2}. To the best of our knowledge, the only probability-competitive algorithms for general matroids are the $\Omega({1/\log^2{r}})$-competitive algorithm by Bateni, Hajiaghayi, and Zadimoghaddam \cite{bateni-secretary} and a $\Omega({1/\log{r}})$-competitive algorithm by Soto, Turkieltaub, and Verdugo~\cite{forbidden-paper}.

The matroid secretary conjecture has been shown to hold for the \emph{random assignment and random order} model by Soto \cite{soto-secretary}, a setting with a slightly weaker adversary in which the weights of the elements are still chosen adversarially, but the assignment of weights to elements is done uniformly at random. Oveis-Gharan and Vondrák \cite{GharanV13}
extend this result to the setting where the arrival order of the elements is adversarial instead of uniformly random. 
Santiago, Sergeev, and Zenklusen \cite{SantiagoSZ23} circumvent this problem by obtaining an algorithm that only requires access to an independence oracle in the random-assignment random-order model. Recently, Cristi et al. \cite{matroid-embeddings} gave reductions from the unknown-matroid to the known-matroid variants of MSP for several sub-classes of matroids and showed an impossibility result for such a reduction in the case of a general matroid.
A series of works by Dughmi \cite{dughmi1,dughmi2} uses duality to establish an equivalence between the MSP conjecture and the existence of random-order contention resolution schemes. 
Babaioff et al.~\cite{babaioff-secretary-2} unified several results via the $\alpha$-partition property technique.\\ 
%Soto, Turkieltaub, and Verdugo \cite{forbidden-paper} introduced the forbidden-sets technique and obtained many of the current best competitive ratios for several sub-classes of matroids. \\

\noindent{\it Laminar Matroids.}
The initial $3/16000$ guarantee of Im and Wang \cite{ImW11} was beaten by Jaillet, Soto, and Zenklusen, based on the $\alpha$-partition property, that is $1/(3 e \sqrt{3}) \approx 0.0707$-competitive \cite{free-order-secretary}. Later, Ma, Tang, and Wang \cite{ma-tang-secretary} showed that the greedy-improving algorithm is $0.1041$-competitive, a guarantee that was later improved by Soto, Turkieltaub, and Verdugo \cite{forbidden-paper} using the forbidden sets technique to show that laminar matroids are $3$-forbidden, yielding a $0.1924$-probability-competitive algorithm. 
Prior to our work, the state-of-the-art bound was the $0.2105$-competitiveness for all laminar matroids shown by Huang, Parsaeian, and Zhu \cite{zahra-laminar-secretary}.\\

\noindent{\it Graphic Matroids.}
The initial paper on the MSP by Babaioff et al.~\cite{mat-sec} showed a $1/16$-competitive algorithm for graphic matroids, and this bound was improved to $1/(3e)$ by Babaioff et al.~\cite{babaioff-secretary-2}. Both of these results rely on the $\alpha$-partition property. This approach was followed by Korula and Pal \cite{korula-pal-graphic}, who showed that graphic matroids have the $2$-partition property, improving the bound to $1/(2e)$ using a simple algorithm. Finally, Soto et al. \cite{forbidden-paper} showed that graphic matroids are $2$-forbidden, yielding a $1/4$-probability-competitive algorithm.
Concurrently with our work, and using a different approach, Banihashem et al.~\cite{banihashem2025beating} also show how to surpass $0.25$ for graphic matroids, to get $0.2652$ for simple graphs, and a $0.2531$ factor for general graphs.\\

\noindent{\it Impossibility Results.} Babaioff et al.~\cite{mat-sec} showed that the greedy-improving algorithm alone cannot be constant-competitive for graphic matroids. Then, Bahrani et al.~\cite{BahraniBSW21} generalized this result, showing that this is also the case for a bigger class of greedy-like algorithms and for algorithms that partition the ground set without looking at the weights, and then work on each part separately. 
Recently, Abdolazimi et al.~\cite{matroid-partition} showed that algorithms based on the $\alpha$-partition property cannot be constant-competitive for the full binary matroid, even if the partition obtained can depend on the initial set of sampled elements.
%%%%%%%%%%%%%%%%
\section{Preliminaries}\label{sec:prelim}

Given a ground set $E$ together with a subset $F\subseteq E$ and $e\in E$, the sets $F\setminus \{e\}$ and $F\cup\{e\}$ are abbreviated as $F-e$ and $F+e$, respectively.  
We denote the reals, non-negative reals, and non-negative integers by $\RR$, $\RR_+$, and $\NN$, respectively. 
For $n\in \NN$, we use $[n]$ to denote the set $\{1,\dots, n\}$. For a set $S\subseteq \NN$, we denote by $S^*$ the set of finite sequences whose entries are elements of $S$ and refer to these sequences as {\it words}. 
The set of words of length $k$ is denoted by $S^k$. 
We use $\epsilon$ to denote the empty word and $xy$ to denote the {\it concatenation} of two words $x$ and $y$. 
We also make use of standard shorthand notation for regular languages such as $w=1^a2^b$ to denote the word $11\dots1 22\dots2$ with $a$ many $1$s and $b$ many $2$s, or $w \in 1^*2$ to denote that $w \in \{1^k2\colon k\in \NN\}$.\\
\begin{comment}
Let $D = (V,A)$ be a {\it directed graph} with vertex set $V$ and arc set $A$. 
We denote an arc $a$ oriented from $u$ to $v$ by $(u,v)$ (or sometimes $uv$ for brevity), where $u$ and $v$ are called the {\it tail} and the {\it head} of $a$, respectively. 
An arc $uv$ {\it enters} a subset $Z$ of vertices if $v \in Z$, $u \notin Z$ and {\it exits} $Z$ if $u \in Z$, $v \notin Z$. For a subset $F \subseteq A$ of arcs, the {\it in-degree} of $Z$ is the number of arcs entering it and is denoted by $\deg^-_F(Z)$. The subscript $F$ is dismissed when $F$ consists of the whole arc set. 
\end{comment}

\noindent{\bf Matroids.} We provide a few basic definitions for matroids and refer the reader to~\cite{oxley2011matroid} for an extensive treatment. 
A {\it matroid} $\calM = (E, \calI)$ is defined by its {\it ground set} $E$ and its {\it family of independent sets} $\calI\subseteq 2^E$ that satisfies the so-called {\it independence axioms}: (I1) $\emptyset\in\calI$, (I2) $X\subseteq Y,\ Y\in\calI\Rightarrow X\in\calI$, and (I3) $X,Y\in\calI,\ |X|<|Y|\Rightarrow\exists e\in Y-X\ \text{such that } X+e\in\calI$. 
For $X\subseteq E$, the maximum size of an independent set in $X$ is called the {\it rank} of $X$. 
A non-independent set is called {\it dependent}, and a {\it circuit} is an inclusion-wise minimal dependent set. 
Two elements $e,f\in E$ are {\it parallel} if they form a circuit of size two. 
The independence axioms imply that any two maximal independent sets in $E$ have the same size, called the {\it rank of the matroid}, usually denoted by $r_\cM$. 

The {\it rank function} of the matroid is a function $r_\cM: 2^E \to \NN$, where $r_\cM(S)$ denotes the size of a maximal independent set contained in $S$. In both cases, the subscript is dropped if the matroid is clear from context. The {\it bases} of the matroid are its maximal independent sets. An element $e\in E$ is a {\it loop} if $\{e\}$ is a circuit and a {\it coloop} if $e$ is contained in every base. Throughout the paper, we use $n$ and $r$ to denote the size of the ground set and the total rank of the matroid, respectively, i.e., $n = |E|$, and $r=r_\cM(E)$.\\

\noindent{\bf Matroid Secretary Problem (MSP).}
In the MSP, one is given a matroid $\mathcal{M} = (E, \mathcal{I})$ and is presented with the elements of $E$ in an online manner in a uniformly random order. Upon observing an element $e \in E$, one must decide immediately and irrevocably whether to select $e$ or reject it and continue to the next element, subject to the constraint that the selected set $S$ must be independent in $\mathcal{M}$, i.e., $S \in \mathcal{I}$. It is assumed throughout that the matroid description is given upfront to the algorithm and that the matroid is loopless, as loops can be skipped during the online process.
%In the {\it utility MSP}, an injective weight function $w: E \to \mathbb{R}_+$ is given, which is not known a priori but is revealed upon the arrival of each element. The objective is to select a set $S$ that maximizes the total weight $w(S) = \sum_{e \in S} w(e)$. An algorithm $\ALG$ selecting a (random) set $S$ is said to be {\it $\alpha$-utility-competitive}, with $\alpha \in [0,1]$, if $\mathbb{E}[w(S)] \geq \alpha \cdot\OPT(E)$, where $\mathrm{OPT}(E)$ denotes the maximum weight of an independent set in $\mathcal{M}$ with respect to $w$.

In this work, we consider the general ordinal setting for MSP, i.e., instead of a weight function, a total order is given on the elements of the matroid and denoted by $\succ$. 
Note that, by the greedy algorithm, any subset $S \subseteq E$ admits a unique independent set that is lexicographically maximal with respect to $\succ$, which we denote by $\OPT(S)$. Upon the arrival of a new element $e$, its relative rank is revealed by comparing it with the elements that have already arrived. 
An algorithm $\ALG$ selecting a (random) set $S$ is said to be {\it $\alpha$-probability-competitive} if every element of $\OPT(E)$, the lexicographically maximal independent set of $E$ with respect to $\succ$, is included in $S$ with probability at least $\alpha$. The elements in $\OPT(E)$ are referred to as \textit{optimal elements}.

%%%%%%%%%%%%%%%%
\section{The Labeling Scheme Framework}
\label{sec:label}
%%%%%%%%%%%%%%%%

%%%%%%%%%%%%%%%
\subsection{Improving Elements and Times}
\label{sec:improving_defs}
%%%%%%%%%%%%%%%
In this section, we describe our new approach to the matroid secretary problem using labeling schemes.
Given a ground set $E$ together with a subset $F\subseteq E$ and $e\in E$, the sets $F\setminus \{e\}$ and $F\cup\{e\}$ are abbreviated as $F-e$ and $F+e$, respectively.
We use $\cM=(E,\cI)$ to denote the matroid and $r$ to indicate its rank. The elements of $E$ can be compared (there is an underlying total order), or equivalently, we assume that all the elements have different weights. For a subset $X \subseteq E$, we denote by $\OPT(X)$ the unique optimal basis of $X$ and refer to the elements in $\OPT(E)$ as {\it optimal elements}.

An equivalent way to model a uniformly random arrival order is to assume that every element $e \in E$  picks an arrival time $t_e$ from the uniform distribution over $[0,1]$, and the algorithm observes elements one by one in increasing arrival time. Although arrival times are not part of the input, our algorithms for matroids with $n$
elements can simulate this process by first generating 
$n$
independent, uniformly distributed arrival times in 
$[0,1]$. These \emph{arrival times} are sorted as $s_1 < \dots < s_n$ and assigned to elements as they are being presented; if $e$ is the $i$-th arriving element, then $t_e = s_i$.
For every positive $t \in (0,1]$, let $E_t$ denote the set of elements arriving in the interval $[0,t)$ and $E_t^+$ the set of elements arriving in the interval $[0,t]$. 

\begin{definition}[Improving Elements and Improving Times]
An element $e \in E$ is improving if $e\in \OPT(E_{t_e}^+)=\OPT(E_{t_e}+e)$. The times when improving elements arrive are the {\it improving times}.
\end{definition}

Contrary to arrival times, we usually assume they are sorted in decreasing order, i.e., for a set of improving times $\set{t_1, t_2, \dots, t_k}$, we have $t_k < t_{k-1} < \dots < t_1$. For $0 \leq a < b \leq 1$, we denote the {\it last improving time} in increasing order in $[0,b)$ by $S(b)$, and the {\it number of improving times} in $[a,b)$ or $[a,b]$ by $N[a,b)$ and $N[a,b]$, respectively.
The following is a folklore lemma for matroids (see, e.g.,~\cite[Lemma 2]{forbidden-paper}). 
It allows us to restrict our analysis to algorithms for the matroid secretary problem that always skip non-improving elements.
\begin{lemma}\label{lem:improving-elems}
If $X\subseteq Y$ and $e\in \OPT(Y+e)$, then $e\in \OPT(X+e)$. In particular, every optimal element $e \in \OPT(E)$ is improving.
\end{lemma}

Our algorithms skip all elements with arrival time before a certain parameter $p \in (0,1)$, which has the same effect as choosing a value $n_0\sim \Bin(n,p)$ from the binomial distribution of parameters $n$ and $p$ and then skipping the first $n_0$ elements. Among the elements that arrive after time $p$, our algorithms skip all non-improving elements. In particular, the output of all our algorithms will be a subset of the improving elements arriving in the interval $[p,1]$.
The following assumption is unnecessary for our algorithms, but it simplifies their analysis.

\begin{assumption}\label{asm:copies}
The input matroid is augmented in such a way that, at every time $t$, the set $E_t$ has full rank. Every new element added in this way has an infinitesimal weight compared to its original counterpart.
\end{assumption}
We remark that this is only used to analyze our algorithms, and we refer to the full version of this paper for more details about the construction. One of the main tools we use is the following lemma, which characterizes the distribution of the arrival times of improving elements. 
Its proof can be found in the full version of this paper.

\begin{lemma}\label{lem:Poisson-all}
Under Assumption~\ref{asm:copies}, the following hold:
\begin{enumerate}[itemsep=0em,label=\normalfont(\roman*)]

\item \label{it:1} For every $b\in\RR_+$, $S(b)$ has cumulative distribution function $\Pr\brk{S(b)\leq x}=0$ if $x<0$, 
$\Pr\brk{S(b)\leq x}=\prn{\frac{x}{b}}^r$ if $x\in [0,b)$, and  $\Pr\brk{S(b)\leq x}=0$ if $x\geq b$.

\item \label{it:2} For all $k\in\NN$, let $x_k = -\ln(y_k)$ where $y_0 = 1$ and $y_k = S(y_{k-1})$ otherwise.
\begin{enumerate}[label=\normalfont(\alph*)]\itemsep0em
 \item \label{it:a} The arrival times $x_k - x_{k-1}$ are distributed according to an exponential distribution with mean ${1}/{r}$. In particular, $\{x_k\}_{k\in \NN}$ is a Poisson point process restricted to the positive half-line with rate $r$. 
 \item \label{it:b} $N[a,b)$ distributes as a Poisson random variable with rate $\lambda = r \ln(b/a)$.
 \item \label{it:c} For any finite family of disjoint intervals $\{[a_i,b_i)\}_{i\in I}$, the random variables $\{N[a_i,b_i)\}_{i\in I}$ are mutually independent.
\end{enumerate}
\end{enumerate}
\end{lemma}

\begin{proof}
\ref{it:1}. 
We follow the approach of~\cite{zahra-laminar-secretary}. 
Recall that $E_b$ is the set of elements arriving in $[0,b)$. Conditioned on the set $E_b=X$, the last improving element in $[0,b)$ is a uniformly chosen random element of $\OPT(E_b)$. Therefore, its arrival time is $S(b)=\max\{t_e \colon e\in \OPT(X)\}$. First, notice that, by definition, $S(b)\in [0,b)$. Also, since the arrival times are independent,  
\[
\Pr\brk{S(b) \leq x} = \prod_{e\in \OPT(X)}\Pr(t_e\leq x) = \prn{\frac{x}{b}}^r,
\]
for every $x \in [0,b)$.
\medskip

\noindent \ref{it:2}. For condition~\ref{it:a}, notice that for $x \in [0, +\infty)$ we have
\begin{align*}
\Pr[x_k - x_{k-1} \leq x] &= \Pr[-\ln{y_k} + \ln{y_{k-1}} \leq x] = \Pr\left[\ln\prn{\frac{y_{k-1}}{y_k}} \leq x\right] \\
&= \Pr\left[\ln\prn{\frac{y_{k-1}}{S(y_{k-1})}} \leq x\right] = \Pr\left[S(y_{k-1}) \geq y_{k-1} e^{-x}\right] \\
&= 1 - \Pr\left[S(y_{k-1}) < y_{k-1} e^{-x}\right] = 1 - \prn{\frac{y_{k-1} e^{-x}}{y_{k-1}}}^r= 1 - e^{-xr}.
\end{align*}
In other words, for every positive $k\in\NN$, $x_k - x_{k-1}$ is an exponential random variable of rate $r$. We conclude that $\{x_k\}_{k\in \NN}$ is a Poisson point process with rate $r$ on the positive half-line.

Condition~\ref{it:b} holds since $N[a,b)$ counts the number of points in $\{y_k\}_{k\in \NN}$ that lie in $[a,b)$, or equivalently, the number of points in the Poisson process $\{x_k\}_{k\in \NN}$ that lie in $(-\ln(b),-\ln(a)]$, and this behaves like a Poisson variable of rate $r$ multiplied by the length of the interval $(-\ln(b),-\ln(a)]$, i.e., like a Poisson with rate $r \ln\prn{\frac{b}{a}}$.
Finally, condition~\ref{it:c} follows from the complete independence of the Poisson process.
\end{proof}
%%%%%%%%%%%%%%%%
\subsection{Labeling Schemes and Improving Words}
\label{sec:improving}
%%%%%%%%%%%%%%%%

Our main contribution is a new way of analyzing algorithms based on assigning labels from $[r]$ to the elements of $\OPT$ at improving times and then analyzing the {\it word} formed by the labels of the improving elements.

\begin{definition}[Labeling Scheme]
A {\it labeling scheme $\Lambda$ for an interval $[a,b]$} assigns, for each improving time $t \in [a,b]$, a unique label from $[r]$ to each element of the set $X = \OPT(E_t^+)$.
\end{definition}

It is worth emphasizing that the labels are assigned independently of the arrival order of $E_t$. Also note that an element's label may change from improving time to improving time and can depend on the set of improving elements and times in $(t,1]$. We often use the following labeling scheme.

\begin{definition}[Labeling Scheme Induced by a Total Order] \label{def:inducedlabeling}
Let $\pi$ be any total order on $E$. For every $t \in [0,1]$, the {\it labeling scheme $\Lambda_\pi$ induced by $\pi$} assigns a label to each $e \in \OPT(E_t^+)$ equal to the relative rank of $e$ in $\OPT(E_t^+)$ with respect to $\pi$.
\end{definition}
In other words, for $X = \OPT(E_t^+)$, the labeling scheme assigns label $1$ to the element of $X$ that appears first in $\pi$, label $2$ to the element of $X$ that appears next, and so on. Note that even for this simple labeling scheme, the same element may have different labels at different times. Given a labeling scheme $\Lambda$, we say that an {\it improving time $t$ has label $i$}, denoted as $z(t) = i$ if the improving element arriving at time $t$ receives label $i$. We define one more object that is useful in our analysis.

\begin{definition}[Improving Word in Time Interval {$[a,b]$}]
Let $0 \leq a < b \leq 1$ and $\set{e_k, \dots, e_1}$ be the set of improving elements that arrive in time interval $[a,b]$ with improving times $t_j= t_{e_j}$, where the elements are indexed in a decreasing order with respect to their arrival times, i.e., $t_k < \dots < t_1$. The {\it improving word in $[a,b]$} is $z = z_1 \dots z_k \in [r]^k$, where $z_i=z(t_i)$ is the {\it label of $e_i$ in $\OPT(E_{t_i}^+)$}. 
\end{definition}
For an illustration of the previous definition, see Fig.~\ref{fig:improving-word}.
Whenever we apply a labeling scheme $\Lambda$ and the resulting improving word in an interval $[a,b]$ is $z$, we say that \emph{$\Lambda$ produces $z$ in $[a,b]$}. The uniformly random arrival order of the elements implies the labels of $z$ are chosen uniformly at random from $[r]$, and the length of $z$ follows a Poisson distribution. The main idea of our analysis is to design a language $\cL$, depending on the matroid class and rank, such that, for every optimal element $e^* \in \OPT(E)$, we can find a labeling scheme $\Lambda$ which ensures that $e^*$ is selected by the algorithm whenever the improving word produced by $\Lambda$ in $[p,1]$ belongs to $\cL$. We call such a language an \emph{improving language}. 
%In the full version of this paper, we state and prove some technical results regarding labeling schemes and improving words used in analyzing our algorithms.
\sidecaptionvpos{figure}{c}
\begin{SCfigure}[50][t!]
\centering
\includegraphics[width=0.5\textwidth]{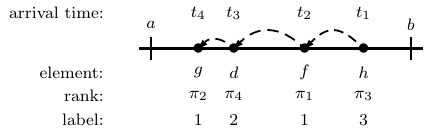}
\caption{A sequence of improving element arrivals in $[a,b]$. Assume that $E_{t_4}=\{g\}$, $E_{t_3}=\{g,d\}$, $E_{t_2}=\{g,d,f\}$, and $E_{t_1}=\{g,d,f,h\}$. If the elements are ranked as $f >_{\pi} g >_{\pi} > h >_{\pi} d$, the improving word induced by $\pi$ in $[a,b]$ is $z = 3121$.}
\label{fig:improving-word}
\end{SCfigure}

\subsection{Technical Tools for Analyzing Labeling Schemes}\label{sec:technical-labeling}
We present some technical results regarding labeling schemes and improving words used in analyzing our algorithms.
For a set $S$ of labels, we denote by $N_S[a,b)$ and $N_S[a,b]$ the number of improving times in $[a,b)$, and $[a,b]$ respectively, for which the label belongs in $S$.

\begin{lemma}\label{lem:Poisson-label}
Consider a labeling scheme $\Lambda$. Let $S \subseteq [r]$ be a subset of labels and $t \in (0,1]$, and $\Lambda(e_t)$ denote the label assigned by $\Lambda$ to an element arriving at time $t$.
\begin{enumerate}[itemsep=0pt,label=\normalfont(\roman*)]
    \item \label{poi:1} If $x < t$ is an improving time, then
    $\Pr[\Lambda(e_x) \in S] = \frac{|S|}{r}.$
    Furthermore, this event is independent of the label of any improving time in $(x, t)$.
    
    \item \label{poi:2} The set $\{-\ln(s) \colon s\in T_S(t)\}$ is a Poisson point process in $(-\ln(t),\infty)$. In particular, for every $0<a< b\leq t$, $N_S[a,b)$ is distributed as a Poisson random variable with rate $\lambda = |S| \ln\prn{\frac{b}{a}}$.
    \item \label{poi:3} For any finite family of disjoint intervals $\{[a_i,b_i)\}_i$ and arbitrary -- not necessarily disjoint -- subsets $\{S_i\}_i$ of $[r]$, the random variables $\set{N_{S_i}[a_i,b_i)}_i$ are mutually independent.
    
    \item \label{poi:4} For any family of disjoint subsets $\{S_i\}_i$ of $[r]$ and arbitrary -- not necessarily disjoint -- intervals $\{[a_i,b_i)\}_i$, the random variables $\{N_{S_i}[a_i,b_i)\}_i$ are mutually independent.
\end{enumerate} 
\end{lemma}
\begin{proof}
\ref{poi:1}. Observe that for any improving time $x$, conditioned on the set $X = E_x^+$ of elements that have arrived in the interval $[0,x]$, the element that arrives at time $x$ is chosen uniformly at random from $\OPT(E_x^+)$. Since $\Lambda$ is independent of the internal arrival ordering of $E_x^+$, the label of $x$ is uniformly chosen from $[r]$. Therefore, the probability that the label is in $S$ is exactly $|S| / r$. \medskip

\noindent \ref{poi:2} -- \ref{poi:4}. Each statement is followed directly by the coloring property of Poisson point processes; see~\cite[Pages 163--164]{borovkov2003elements} for further details.
\end{proof}

\begin{lemma}\label{lem:words}
Let $z$ be the random variable representing the improving word produced in the interval  $[a,b]$, $|z|$ be its length, and $\lambda= r \ln\prn{\frac{b}{a}}$. Furthermore, let $S, T\subseteq [r]$ be two disjoint subsets of the labels, $z_S$ (resp. $z_T$) be the subsequence of $z$ obtained by deleting all symbols not in $S$ (resp. not in $T$), and let $\lambda_S = |S|\ln\prn{\frac{b}{a}}$.
\begin{enumerate}[itemsep=0pt,label=\normalfont(\roman*)]
    \item For every $k\in \NN$, $\Pr[ |z| = k ] = \Pr[N[a,b)=k] = \frac{\lambda^ke^{-\lambda}}{k!}$.
    \item Conditioned on $|z|=k$, $z$ is a uniform random word in $[r]^k$.

    \item For every $k\in \NN$, $\Pr[ |z_S| = k ] = \Pr[N_S[a,b)=k] = \frac{(\lambda_S)^ke^{-\lambda_S}}{k!}$.
    \item Conditioned on $|z_S|=k$, $z_S$ is a uniform random word in $S^k$.
    \item $z_S$ and $z_T$ are independent.
\end{enumerate}
\end{lemma}

\begin{proof}
Direct consequence of Lemma~\ref{lem:Poisson-label} when applied to improving words.
\end{proof}
Finally, the following lemma is useful in characterizing the improving words that guarantee selecting a specific element of $\OPT(E)$ in our algorithms.

\begin{lemma}\label{lem:standardlabel}
Let $e^*\in \OPT(E)$ and $\pi$ be any total order of $E$ for which $\pi_1 = e^*$. Let $p \leq t_k <  \dots < t_1 \leq 1$ be the improving times in $[p,1]$ and $z = z_1 z_2 \dots z_k$ be the associated improving word in $[p,1]$ with respect to the labeling $\Lambda_\pi$.  
\begin{enumerate}[itemsep=0pt,label=\normalfont(\roman*)]
    \item \label{stdlbl:1} $z$ contains a symbol $1$ if and only if the arrival time of $e^*$ is in $[p,1]$.
    
    \item \label{stdlbl:2} If $z$ contains a symbol $1$ and $j = \argmin\{\ell\in[k] \colon z_\ell = 1\}$, then $t_{e^*} = t_j$. 
\end{enumerate}
\end{lemma}
\begin{proof}
\ref{stdlbl:1}. By Lemma~\ref{lem:improving-elems} we know that $e^*$ is improving, and thus $e^* \in \OPT(E_{t_{e^*}}^+)$. Also, $\pi_1 = e^*$, and thus if $t_{e^*} \in [p,1]$, then $\Lambda_\pi$ assigns a label of $1$ to $e^*$ upon arrival and thus $z$ contains a $1$. Furthermore, if $t_{e^*} \in [0,p)$, then for any other improving element $e$ arriving at $t_e \in [p,1]$, the relative rank of $e$ in $\OPT(E_{t_e}^+)$ according to $\pi$ cannot be $1$, since $e^*\in \OPT(E_{t_e}^+)$ and $\pi_1 = e^*$. Thus, $z$ does not contain a $1$.

\noindent \ref{stdlbl:2}. Assume that $z$ contains a $1$. By the first part, this implies that the arrival time of $e^*$ is in $[p,1]$. Since the label of $e^*$ in $\OPT(E^+_{t_{e^*}})$ is $1$, $t_{e^*}$ is among $\{t_\ell \colon \ell \in[k], z_\ell = 1\}$. Suppose that $t_{e^*} = t_i$ for some $i > j = \argmin\{\ell \in [k] \colon z_\ell = 1\}$, and let $e$ denote the improving element arriving at $t_j$. Then, since $t_j > t_i$, both $e^*$ and $e$ are contained in $\OPT(E_{t_j}^+)$. However, this means that the rank of $e$ is at least 2 in $\OPT(E_{t_j})$ yielding a contradiction.
\end{proof}

%%%%%%%%%%%%%%%%
\subsection{Algorithms via Improving Languages}
%%%%%%%%%%%%%%%%
In our framework, after skipping the elements arriving before $p$, we select an improving element if the sequence of different optimum sets we have seen (which correspond to the optimum sets at improving times) satisfies a property $P$, where $P$ is tailored to the specific matroid class and rank $r$. If $P$ holds when an element $e$ arrives, $P$ must guarantee  that $\ALG + e \in \cI$. To analyze such an algorithm, our main goal is to find, for every $e^*\in \OPT(E)$, a labeling scheme $\Lambda:=\Lambda(e^*)$ and a language $\cL \subseteq [r]^*$ such that, if $z \in [r]^*$ is the improving word in $[p,1]$ induced by the labeling scheme $\Lambda$, then $z\in \cL \implies e^*\in \ALG$. The significance of this goal is based on the fact that $\Pr[z \in \cL]$ is much easier to compute due to $z$ coming from a Poisson point process. In the following sections, we exploit this approach for laminar matroids (\Cref{sec:laminar}) and graphic matroids (\Cref{sec:graphic}).
%%%%%%%%%%%%%%%%

%%%%%%%%%%%%%%%%
\section{Laminar Matroids}
\label{sec:laminar}
%%%%%%%%%%%%%%%%

We recall that a matroid $\calM=(E,\calI)$ is {\it laminar} if $\calI=\{F\subseteq E\colon |F\cap E_i|\leq g_i\ \text{for each $i\in [q]$}\}$ for some laminar family $E_1,\dots, E_q$ of subsets of $E$ and upper bounds $g_1,\dots,g_q\in\NN$.
In this section, we apply the labeling framework to obtain an improved analysis of the well-known \textsc{Greedy-Improving} algorithm for laminar matroids. 
In other words, for laminar matroids, we let the property $P$ be that $\ALG + e \in \cI$.

\begin{algorithm}[H]
    \begin{algorithmic}[1]
    \State Initialize $\ALG_0 \gets \emptyset$.
    \State Skip all elements with arrival time in $[0,p)$.
    \For{each element $e$ with $t_e\in [p,1]$, in arrival order}
        \If{$e\in\OPT(t_e)$ and $\ALG_0 + e \in \cI$}\label{st:if}
            \State $\ALG_0 \gets \ALG_0 + e$
        \EndIf
    \EndFor
    \State Return $\ALG_0$.
    \end{algorithmic}
    \caption{\textsc{Greedy-Improving Algorithm}}
    \label{alg:greedy-improving}
\end{algorithm}

As an illustration of our techniques, we first consider, in Section~\ref{sec:uniform}, the special case of uniform matroids and characterize the probability-competitiveness of the algorithm (Proposition~\ref{lem:gi}). Then, we generalize our analysis for the broader class of laminar matroids in Section~\ref{sec:general}, showing that \textsc{Greedy-Improving} is $\prn{1 - \ln(2)}$-probability-competitive for every laminar matroid (Theorem~\ref{thm:laminar-intro}); we prove this bound is tight, thereby settling the probability-competitiveness of \textsc{Greedy-Improving} completely. 
Finally, in Section~\ref{sec:rank2}, we design a new algorithm for rank-$2$ matroids (a subclass of laminar matroids) that outperforms \textsc{Greedy-Improving} to get a $0.3462$-probability-competitiveness. 

We remark that, in the context of laminar matroids, Assumption~\ref{asm:copies} corresponds to augmenting the matroid with a large set $F$ of parallel copies of a single basis $B_0$ (i.e., each element of $B_0$ has a countably infinite number of parallel copies). 
%All missing proofs can be found in the full version of this paper.

%%%%%%%%%%%%%%%%
\subsection{Uniform Matroids}
\label{sec:uniform}
%%%%%%%%%%%%%%%%

As a warm-up, we present a new analysis of \textsc{Greedy-Improving} for the special class of uniform matroids using labeling schemes. It is worth noting that the algorithm is not optimal for uniform matroids; we describe it to aid the reader and illustrate our analysis.
When $(E,\calI)$ is a rank-$r$ uniform matroid on $n$ elements, the condition $\ALG_0+e\in\calI$ in Step~\ref{st:if} is equivalent to $|\ALG_0|\leq r-1$.\\  

%\begin{definition}
\noindent{\bf Labeling Scheme $\Lambda_{\normalfont\text{uniform}}$.}
We now introduce our labeling scheme for the case of uniform matroids.
Let $e^*\in \OPT(E)$. Consider a total order $\pi$ of $E$ in which $e^*$ appears first. We set $\Lambda_{\normalfont\text{uniform}}(e^*)$ to be the labeling scheme induced by $\pi$ as in Definition~\ref{def:inducedlabeling}. 
The following lemma states the correspondence between the improving words produced by this labeling scheme and the elements selected by the greedy-improving algorithm.
%\end{definition}

\begin{lemma}\label{lem:uniform-language}
Let $z$ be the improving word in $[p,1]$ produced by the labeling scheme $\Lambda_{\normalfont\text{uniform}}(e^*)$, $\ALG_0$ be the output of \textsc{Greedy-Improving}, and  $$\mathcal{L}_{\normalfont\text{uniform}}=\{x1y \in [r]^* \colon x\in \{2,3,\ldots,r\}^* \text{ and } |y|\leq r-1\}.$$
Then $e^*\in \ALG_0$ if and only if $z\in \mathcal{L}_{\normalfont\text{uniform}}$.
\end{lemma}

\begin{proof}
By Lemma~\ref{lem:standardlabel}, $e^*\in \ALG_0$ if and only if $z$ contains a symbol 1 and the first appearance of $1$ in $z$ is on one of the last $r$ positions of $z$. This is equivalent to $z\in \mathcal{L}_{\text{uniform}}$.
\end{proof}

Let $c(r,p) = -p\ln(p)$ if $r=1$ and
\begin{align*}
c(r,p) = \frac{1-\prn{1-\frac{1}{r}}^r}{\prn{1-\frac{1}{r}}^r}p + p^r\sum_{k = 0}^{r-1} \frac{\prn{r\ln\prn{\frac{1}{p}}}^k}{k!}  - \frac{p^r}{\prn{1-\frac{1}{r}}^r}\sum_{k = 0}^{r-1}\frac{\prn{(r-1)\ln\prn{\frac{1}{p}}}^k}{k!},
\end{align*}
for $r\geq 2$. 
In the following proposition, we give an exact formula for the competitiveness of greedy-improving in uniform matroids as a function of the rank $r$.
This formula is easy to evaluate for small $r$ -- see Table~\ref{tab:GreedyUniform} in Appendix~\ref{app:tables}.

\begin{proposition}\label{lem:gi}
The \textsc{Greedy-Improving} Algorithm~\ref{alg:greedy-improving} is exactly $c(r,p)$-probability-competitive for rank-$r$ uniform matroids. Furthermore, as $r\to \infty$, the best sample size for uniform matroids is $p=1/e$, achieving a probability-competitiveness of $1 - 1/e$.
\end{proposition}

\begin{proof}
By Lemma \ref{lem:uniform-language},
\[
\Pr[e^*\in \ALG_0] = \Pr[z=x1y \text{ with $x\in \{2,3,\ldots,r\}^*, |y|\leq r-1$}].
\]
Let $\lambda = r\ln\prn{{1}/{p}} = -r\ln(p)$ and let us apply Lemma~\ref{lem:words}. For $r=1$, we get~$\Pr[e^*\in \ALG_0] = \Pr[z=1] = \frac{\lambda^1e^{-\lambda}}{1!}= -p\ln p.$ 
For $r\geq 2$, we have
\begin{align*}
&\Pr[e^*\in \ALG_0]\\
&\quad=\sum_{k=1}^\infty \sum_{m=0}^{\min(k,r)-1}  \Pr[z=x1y, |z|=k, |y|=m, \text{ with $x\in (\{2,3,\ldots,r\})^{k-m-1}$}]\\
&\quad=\sum_{k=1}^\infty \frac{\lambda^k e^{-\lambda}}{k!}\sum_{m=0}^{\min(k,r)-1} \prn{1-\frac{1}{r}}^{k-m-1} \cdot \frac{1}{r} \\
&\quad=\sum_{k=1}^\infty \frac{\lambda^k e^{-\lambda}}{k!} \prn{  \prn{1-\frac{1}{r}}^{k-\min(k,r)} - \prn{1-\frac{1}{r}}^k } \\
&\quad=\sum_{k=0}^{r-1} \frac{\lambda^k e^{-\lambda}}{k!} \prn{1 - \prn{1-\frac{1}{r}}^k } + \sum_{k = r}^{\infty} \frac{\lambda^k e^{-\lambda}}{k!} \prn{\prn{1-\frac{1}{r}}^{k-r} - \prn{1-\frac{1}{r}}^k}.
\end{align*}
The second equality holds since the probability that $z_{k-m} = 1$ is ${1}/{r}$ and the probability that $x \in \{2,3,\ldots,r\}^{k-m-1}$ is $\prn{1-{1}/{r}}^{k-m-1}$.

Let $\lambda' = \lambda\prn{1-{1}/{r}} = (r-1)\ln\prn{{1}/{p}}$ and note that $\lambda - \lambda' = {\lambda}/{r} = \ln\prn{{1}/{p}}$. Denote by $P(\mu)$ a random variable distributed as a Poisson with parameter $\mu$ for any $\mu$. Then,
\begin{align*}
&\Pr[e^*\in \ALG_0] \\
&\quad= \Pr[P(\lambda)\leq r-1]- \sum_{k=0}^{r-1}\frac{(\lambda')^ke^{-\lambda}}{k!} + \frac{1 - \prn{1-\frac{1}{r}}^r}{\prn{1-\frac{1}{r}}^r}\sum_{k = r}^\infty \frac{(\lambda')^k e^{-\lambda}}{k!} \\
&\quad=\Pr[P(\lambda) \leq r-1] - e^{-\lambda + \lambda'}\Pr[P(\lambda') \leq r-1]+ \frac{1 - \prn{1-\frac{1}{r}}^r}{\prn{1-\frac{1}{r}}^r}e^{-\lambda+\lambda'}\Pr[P(\lambda')\geq r]\\
&\quad=\Pr[P(\lambda)\leq r-1]- p\Pr[P(\lambda')\leq r-1]+ \frac{1-\prn{1-\frac{1}{r}}^r}{\prn{1-\frac{1}{r}}^r}p(1-\Pr[P(\lambda')\leq r-1])\\
&\quad=\frac{1-\prn{1-\frac{1}{r}}^r}{\prn{1-\frac{1}{r}}^r}p + \Pr[P(\lambda)\leq r-1] - \frac{p}{\prn{1-\frac{1}{r}}^r}\Pr[P(\lambda')\leq r-1] \\
&\quad=\frac{1 - \prn{1-\frac{1}{r}}^r}{\prn{1-\frac{1}{r}}^r} p + p^r\sum_{k = 0}^{r-1} \frac{\prn{r\ln\prn{\frac{1}{p}}}^k}{k!}  - \frac{p^r}{\prn{1-\frac{1}{r}}^r} \sum_{k = 0}^{r-1}\frac{\prn{(r-1) \ln\prn{\frac{1}{p}}}^k}{k!},
\end{align*}
which concludes the proof of the first part of the proposition.

Recall that for $r > 1$,
\[
c(r,p) = \frac{1 - \prn{1-\frac{1}{r}}^r}{\prn{1-\frac{1}{r}}^r} p + \Pr[P(\lambda)< r] - \frac{p}{\prn{1-\frac{1}{r}}^r}\Pr[P(\lambda')< r],
\]
and that the mean and variance of $P(\lambda)$ are both $\lambda$. For $p > {1}/{e}$, we can write ${1}/{\ln\prn{{1}/{p}}} = 1 + \eps$ for some $\eps > 0$, and so $r = \lambda(1+\eps)$. Using Chebyshev's inequality and that $\lambda'< \lambda$, we get
\begin{align*}
\Pr[P(\lambda')<r] &> \Pr[P(\lambda)< r]\\
&= 1 - \Pr[P(\lambda)
\geq \lambda + \lambda \eps] \geq 1 - \frac{1}{\lambda\eps^2} = 1 - \frac{1}{r\ln\prn{{1}/{p}} \eps^2},
\end{align*}
which goes to 1 as $r\to \infty$.
We conclude that, for $p > {1}/{e}$,
\[
\lim_{r\to \infty} c(r,p) = 1+p \lim_{r\to \infty} \prn{\frac{1 - \prn{1-\frac{1}{r}}^r}{\prn{1-\frac{1}{r}}^r} - \frac{1}{\prn{1-\frac{1}{r}}^r}} = 1-p.
\]
On the other hand, if $p < {1}/{e}$ then ${1}/{\ln\prn{{1}/{p}}} = 1-\eps$ for some $\eps > 0$ and $r-1 = \lambda'(1-\eps)$. Then,
\begin{align*}
\Pr[P(\lambda)<r] &< \Pr[P(\lambda')\leq r-1]\\
&=\Pr[P(\lambda')\leq \lambda'- \lambda' \eps] \leq\frac{1}{\lambda'\eps^2} = \frac{1}{(r-1)\ln\prn{{1}/{p}}\eps^2},
\end{align*}
which goes to 0 as $r\to \infty$. We conclude that, for $p < {1}/{e}$,
\[
\lim_{r \to \infty} c(r,p) = p \lim_{r\to \infty} \frac{1-\prn{1-\frac{1}{r}}^r}{\prn{1-\frac{1}{r}}^r} =p (e-1).
\]

By the above, the limit curve for the probability-competitiveness of \textsc{Greedy-Improving} when $r\to \infty$ and $p$ is a piecewise linear continuous function that increases from 0 at $p=0$ to $1 - {1}/{e}$ at $p = {1}/{e}$, and then it decreases linearly to $0$ again at $p = 1$, i.e., the best sample size is ${1}/{e}$.
\end{proof}

%%%%%%%%%%%%%%%
\subsection{General Laminar Matroids}
\label{sec:general}
%%%%%%%%%%%%%%%%

A \textit{laminar matroid} is defined on a ground set with a laminar family of subsets. This family is organized into layers, where the sets in the $i$-th layer are pairwise disjoint, have rank $i$ and sets in lower layers are nested within sets in higher layers. In this section, we show the following theorem.

\begin{theorem}\label{thm:laminar-intro}
The following holds:
\begin{enumerate}[label=\normalfont(\roman*)]
    \item The \textsc{Greedy-Improving} algorithm is $(1 - \ln(2))$-probability-competitive for the MSP on laminar matroids. \label{thm:laminar-lower}
    \item For every $\eps > 0$, there is a laminar matroid $\cM(\eps)$ for which the probability-competitiveness of \textsc{Greedy-Improving} is at most $1 - \ln(2) + \eps$.\label{thm:laminar-upper}
\end{enumerate}    
\end{theorem}

%In order to prove this result, we extend the analysis of \textsc{Greedy-Improving} from uniform to laminar matroids.

%%%%%%%%%%%%%%%%
\subsubsection{Lower Bound: Proof Theorem~\ref{thm:laminar-intro}\ref{thm:laminar-lower}}
\label{sec:lower}
%%%%%%%%%%%%%%%%
For uniform matroids, our analysis relies on a total order $\pi$ of the elements on which $e^*$ appears first. 
In contrast, for the laminar case, we need to consider specific orders that depend on the laminar structure.
The following definition captures the orders we use in our construction.

\begin{definition}[Chain Order]
Let $e^* \in \OPT(E)$ and $C_1 \subseteq \dots \subseteq C_q = E$ be the members of the laminar family containing $e^*$. We refer to these sets as the {\it chain induced by $e^*$}. We say that a total order $(\pi_1,\dots, \pi_n)$ of $E$ is a {\it chain order with respect to $e^*$} if $\pi_1=e^*$ and for $1\leq i<j\leq r$, we have $\{\ell\in[q]\colon \pi_i\in C_\ell\}\subseteq\{\ell\in[q]\colon \pi_j\in C_\ell\}$. In other words, in the enumeration $\pi$, $e^*$ appears first, then the rest of the elements from $E \cap C_1$, then all the elements from $(E \cap C_2) \setminus C_1$, then all the elements from $(E \cap C_3)\setminus C_2$ and so on.\\
\end{definition}

\noindent{\bf Labeling Scheme $\Lambda_{\text{chain}}$.} We now introduce the labeling scheme we use to analyze \textsc{Greedy-Improving} for laminar matroids.
%\begin{definition}
%\noindent{\bf Chain Labeling Scheme for Laminar Matroids.}
Let $e^*\in \OPT(E)$ be an element of the optimal basis and $\pi$ be a chain order with respect to $e^*$. We set $\Lambda^{\pi}_{\text{chain}}(e^*)$ as the labeling scheme induced by $\pi$. Since the choice of $\pi$ is arbitrary, we write $\Lambda_{\text{chain}}(e^*)$ to denote a labeling scheme induced by an arbitrary but fixed chain order $\pi$ with respect to $e^*$.
%\end{definition} 
%\begin{definition}[Well-indexed Word]
We say that a word $z$ is \emph{well-indexed} if it satisfies the following:
\begin{enumerate}[label=\normalfont(\arabic*)]
\item There exists a unique symbol 1 in $z$; i.e. $z = x1y$ with $x, y \in \{2,3,\ldots,r\}^*$.
\item For each $c\in [|y|]$, it holds that
    $\abs{\set{i \in [|y|] \colon y_i \leq c}} \leq c-1.$\label{eq:wi}
   In particular, since $1\leq y_i\leq r$ for all $i$, this condition implies $|y|\leq r-1$.
\end{enumerate}
Denote by $\cL_{\text{laminar}}=\{z\in [r]^*\colon z \text{ is well-indexed}\}$. 
The following lemma states the correspondence between the improving words produced by this labeling scheme and the elements selected by the greedy-improving algorithm.
%\end{definition}

\begin{lemma}\label{lem:well-indexed-alg}
Let $z$ be the improving word of $e^*\in\OPT(E)$ produced by the labeling scheme $\Lambda_{\text{chain}}(e^*)$. 
If $z \in \cL_{\normalfont\text{laminar}}$ then $e^*\in \ALG_0$.
\end{lemma}
\begin{proof}
Let $z$ be the improving word of $e^*$ and suppose that $z=x1y$ is well-indexed, where the symbol 1 refers to the entry of $z$ corresponding to $e^*$. Then, by Lemma \ref{lem:standardlabel}, the arrival time $t_{e^*}$ of $e^*$ is in $[p,1]$ and there are exactly $m= |y|$ improving elements arriving in $[p,t)$. Let us denote these elements  by $f_m, \dots, f_1$ with $t(f_m)<\dots <t(f_1)$, so that the label of $f_i$ is $y_i$ for all $i\in [m]$.

Let $\ALG_- \subseteq \{f_m, \dots, f_{1}\}$ denote the set of elements chosen by the algorithm before the arrival of $e^*$. Take an arbitrary member $C$ of the chain induced by $e$ and let $c = r(C)$ be its rank. Let $s<t_{e^*}$ and sort the elements of $\OPT(E_s)$ increasingly by their labels, i.e., by the chain order $\pi$ with respect to $e^*$. Since all elements of $\OPT(E_s) \cap C$ appear before the elements of $\OPT(E_s) \setminus C$ in this order, we conclude that all elements with label $c+1$ or higher are not in $C$. Since this is true also for all improving times $s$ before $t_{e^*}$, we get that $C \cap \ALG_- \subseteq \{f_i\colon y_i \leq c\}$. As $z$ is well-indexed, we have
\begin{align*}
\abs{C \cap \ALG_-} \leq \{i\in [|y|]\colon  y_i \leq c\} \leq c-1.
\end{align*}
We conclude that $\abs{C \cap \prn{\ALG_- +e^*}} \leq 1+ \abs{C\cap \ALG_-} \leq r(C)$ for each $C$ in the chain induced by $e^*$, therefore $\ALG_- + e^*$ is independent, and thus $e^*$ is added to $\ALG_0$.
\end{proof}

Thus, to lower bound the competitiveness of Algorithm~\ref{alg:greedy-improving}, it suffices to compute the probability that the improving word $z$ produced in $[p,1]$ is well-indexed.

\begin{lemma}\label{lem:good}
Let $m\leq r-1$ and $y\in [r]^{m}$ be a word chosen uniformly at random from $[r]^m$. Then, the probability that condition \ref{eq:wi} holds is equal to $1 - m/r$.
\end{lemma}
\begin{proof}
    Let $m\geq 0,r\geq 1$ be integers. We call a sequence $b\in [r]^m$ {\it good} if it satisfies \ref{eq:wi} and {\it bad} otherwise. Let $B(m,r)$ denote the number of bad sequences. Then, $\Pr\brk{\text{$b$ is good}} = 1-{B(m,r)}/{r^m}.$ We claim that
    \[
    B(m,r) = \begin{cases}
        mr^{m-1} &\text{ if $0\leq m\leq r-1$,}\\
        r^m &\text{ otherwise.}
    \end{cases}\]
If $m\geq r$, then any sequence in $[r]^m$ will contain $m>r-1$ indices that are at most $r$, so every sequence is bad and $B(m,r)=r^m$. Assume now that $m\leq r-1$. The rest of the proof is by induction on $r$. For the base case, observe that if $m=0$, then the empty sequence is good by definition, so $B(0,1)=0$. Suppose now that $r\geq 2$, and let us count the number of bad sequences in $[r]^m$ for some $1\leq m\leq r-1$. Since the length $m$ of the sequence is at most $r-1$, any sequence $b\in [r]^m$ is bad if and only if the subsequence $b'$ of $b$ induced by those indices $j$ for which $b_j\leq r-1$ is bad. Hence
\begin{align*}
B(m,r) &= \sum_{j=0}^m\binom{m}{j}B(j,r-1)\\
&= \sum_{j=1}^m \binom{m}{j}j(r-1)^{j-1}=\sum_{j=1}^m m \binom{m-1}{j-1}(r-1)^{j-1}= m r^{m-1},
\end{align*}
where the first equality follows by the induction hypothesis, and the last one follows from the Binomial Theorem. The lemma then follows since if $m=0$, the probability of selecting a good sequence is 1, and if $1\leq m\leq r$, this probability is $1 - m r^{m-1}/{r^m} = 1 - {m}/{r}$.
\end{proof}

For any $p\in(0,1)$, let us define $a(r,p) = -p\ln(p)$ if $r = 1$ and 
\begin{align*}
    a(r,p) = {}&{}-2p+(2+\ln p)\Pr[P(-r\ln p)<r-1]\\
    {}&{}+2\Pr[P(-r\ln p)=r-1]+ \frac{p}{\prn{1 - \frac{1}{r}}^r}\Pr[P(-(r-1)\ln p)\geq r],
\end{align*}
for $r\geq 2$, where $P(\mu)$ is a Poisson random variable with rate $\mu$.
In the following lemma, we show a lower bound on the competitiveness of the \textsc{Greedy-Improving} algorithm parameterized by the rank $r$ of the laminar matroid.

\begin{lemma}\label{lem:arp}
\textsc{Greedy-Improving} is at least $a(r,p)$-probability-competitive for rank-$r$ laminar matroids.
\end{lemma}

\begin{proof}
Let $e^*\in\OPT(E)$ be any element of the optimal basis and $z$ be the improving word of $e^*$. By Lemma~\ref{lem:well-indexed-alg}, it suffices to show that $a(r,p) \triangleq \Pr[z \text{ is well-indexed}]$. By setting $\lambda = r \ln\prn{{1}/{p}}$, we have
\begin{align*}
    a(r,p)
    &= \sum_{k=1}^{\infty}\frac{\lambda^k e^{-\lambda}}{k!} \Pr[z=x1y \text{ is well-indexed} |\; |z|=k]\\
     &= \sum_{k=1}^{\infty}\sum_{m=0}^{\min(k,r)-1} \frac{\lambda^k e^{-\lambda}}{k!} \Pr[z=x1y \text{ is well-indexed and $|y|=m$} |\; |z|=k]\\
     &=\sum_{k=1}^{\infty}\frac{\lambda^k e^{-\lambda}}{k!}\sum_{m=0}^{\min(k,r)-1} \frac{1}{r} \cdot \prn{1-\frac{1}{r}}^{k-m-1} \prn{1 - \frac{m}{r}},
\end{align*}
since $\prn{1-{1}/{r}}^{k-m-1}$ is the probability that $z$ contains no 1 symbol, ${1}/{r}$ is the probability that the symbol at position $m+1$ from the right is a 1, and $1 - {m}/{r}$ is the probability that $y$ is good by Lemma~\ref{lem:good}. For $r = 1$, this simplifies to
$a(r,p) = \frac{\lambda^1 e^{-\lambda}}{1!} = -p\ln(p).$
For $r\geq 2$, define $\lambda' = {\lambda (r-1)}/{r} = (r-1) \ln\prn{{1}/{p}}$. We obtain that
\begin{align*}
&a(r,p)\\
&= \sum_{k=1}^{\infty}\frac{\lambda^k e^{-\lambda}}{k!}\frac{1}{r}\left(\frac{r-1}{r}\right)^{k-1}\sum_{m=0}^{\min(k,r)-1}\left(\frac{r}{r-1}\right)^{m} \prn{1 - \frac{m}{r}} \\
&=\sum_{k=1}^{\infty}\frac{\lambda^k e^{-\lambda}}{k!}\frac{1}{r}\left(\frac{r-1}{r}\right)^{k-1} \left(2-2r + \left(\frac{r}{r-1}\right)^{\min(k,r)-1}(2r-\min(k,r))\right)\\
&=\frac{(2-2r)}{r-1}e^{-\lambda+\lambda'}\sum_{k=1}^{\infty}\frac{(\lambda')^{k} e^{-\lambda'}}{k!}+ \sum_{k=1}^{r-1}\frac{\lambda^{k} e^{-\lambda}}{k!}\frac{1}{r}(2r-k) + \sum_{k=r}^{\infty}\frac{\lambda^{k} e^{-\lambda}}{k!}\left(\frac{r-1}{r}\right)^{k-r} \\
&=-2e^{-\lambda+\lambda'}(1-e^{-\lambda'}) + 2\sum_{k=1}^{r-1}\frac{\lambda^{k} e^{-\lambda}}{k!}-\sum_{k=1}^{r-1}\frac{\lambda^{k} e^{-\lambda}}{k!}\frac{k}{r}+\left(\frac{r-1}{r}\right)^{-r}e^{-\lambda+\lambda'}\sum_{k=r}^{\infty}\frac{(\lambda')^{k} e^{-\lambda'}}{k!} \\
&=-2e^{-\lambda+\lambda'}+2e^{-\lambda} + 2\sum_{k=1}^{r-1}\frac{\lambda^{k} e^{-\lambda}}{k!}-\frac{\lambda}{r}\sum_{k=1}^{r-1}\frac{\lambda^{k-1} e^{-\lambda}}{(k-1)!}+\left(\frac{r-1}{r}\right)^{-r}e^{-\lambda+\lambda'}\sum_{k=r}^{\infty}\frac{(\lambda')^{k} e^{-\lambda'}}{k!}\\
&=-2e^{-{\lambda}/{r}}+2\Pr[P(\lambda)<r]-\frac{\lambda}{r}\Pr[P(\lambda)<r-1]+\left(\frac{r-1}{r}\right)^{-r}e^{-{\lambda}/{r}}\Pr[P(\lambda')\geq r].
\end{align*}
Using that $\lambda = r \ln\prn{{1}/{p}}$, we conclude
\[
a(r,p) = -2p+(2+\ln p) \Pr[P(\lambda)<r-1] + 2\Pr[P(\lambda)=r-1] + \frac{p}{\prn{1-\frac{1}{r}}^r}\Pr[P(\lambda')\geq r],    
\]
which finishes the proof of the lemma.
\end{proof}

Using this lemma, we obtain the lower bound in Theorem~\ref{thm:laminar-intro} by studying the function $a(r,p)$. One can evaluate $a(r,p)$ for small $r$ -- see Table~\ref{tab:table-laminar} in Appendix~\ref{app:tables}. We are now ready to prove the lower bound in Theorem~\ref{thm:laminar-intro}.

\begin{proof}[Proof of Theorem~\ref{thm:laminar-intro}\ref{thm:laminar-lower}]
We use the following claim.

\begin{claim}\label{lem:limit}
For $p\in ({1}/{e},1)$, we have $\lim_{r\to\infty} a(r,p)=2-2p+\ln(p)$.
For $p\in(0,{1}/{e})$, we have $\lim_{r\to\infty} a(r,p)=(e-2)p$.
\end{claim}

\begin{proof}
In the proof of Proposition~\ref{lem:gi}, we showed that as $r\to \infty$, $\Pr[P(\lambda)<r]$ and $\Pr[P(\lambda')<r]$ tend to 1 if $p>1/e$, and to 0 if $p<1/e$. Since $\max_{k\in\NN}\Pr[P(\lambda)=k]$ tends to 0 as $\lambda\to \infty$, we conclude that for $p>1/e$,  $a(r,p)$ tends to $-2p+(2+\ln(p))\cdot 1 + 2\cdot 0 + (ep/e^{-1})\cdot 0=2-2p+\ln(p)$ as $r\to\infty$; and for $p<1/r$, $a(r,p)$ tends to $-2p+(2+\ln(p))\cdot 0 + 2\cdot 0 + (p/e^{-1})\cdot 1=p(e-2)$. This ends the proof of the claim.
\end{proof}

Since any algorithm for rank-$(r+1)$ laminar matroids can be directly applied to rank-$r$ laminar matroids by simply adding a dummy co-loop to the matroid, we deduce that for every rank $r$, and every value $p\in (0,1]$, the competitive ratio of Algorithm \ref{alg:greedy-improving}  is at least $\sup_{r'\geq r} a(r',p)\geq \lim_{r'\to \infty} a(r',p)$.

Noting that for all $p<1/e$, $p(e-2)<1-2/e$, and this value is exactly $2-2x+\ln x$ for $x=1/e$, we deduce that the best single value $p$ we can use to maximize the competitive ratio on all laminar matroids is the one that maximizes $2-2p+\ln p$. This expression is maximized on $p={1}/{2}$, achieving a value of $1-\ln 2\approx 0.3068$.
\end{proof}
%%%%%%%%%%%%%%%%
\subsubsection{Upper Bound: Proof of Theorem~\ref{thm:laminar-intro}\ref{thm:laminar-upper}}
\label{sec:upper}
%%%%%%%%%%%%%%%%
We show that the analysis of Algorithm~\ref{alg:greedy-improving} is tight for laminar matroids.
For $q,r\in\NN$, we define a laminar matroid $\cM_{q,r} = (E_{q,r}, \cI_{q,r})$ together with a weight function as follows. The ground set of the matroid is $E_{q,r}=F_1 \sqcup \dots \sqcup F_{r}$, where each $F_i$ has cardinality $q$. Let $E_i= \cup_{j = 1}^i F_j$. The independent sets are $$\cI_{q,r}=\set{I \subseteq E_{q,r} : \abs{I \cap E_i} \leq i\ \text{for all}\ i \in [r]}.$$ Finally, for each $i\in[r]$, we set the weight of each element $e\in F_i$ to be $w_e \sim U[r-i, r - i + 1]$ independently from all other elements (see Fig.~\ref{fig:tight-matroid}). Choosing $q$ and $r$ carefully results in an instance that provides an upper bound on the probability-competitiveness of \textsc{Greedy-Improving} for laminar matroids.
\sidecaptionvpos{figure}{c}
\begin{SCfigure}[50][t]
\centering
\includegraphics[width=0.6\textwidth]{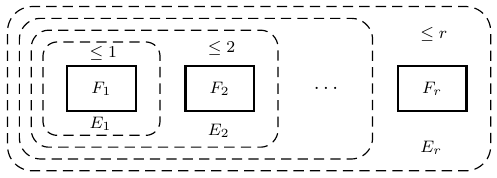}
\caption{Tight example for Alg.~\ref{alg:greedy-improving}. Sets $F_1,\dots,F_r$ are presented with solid lines; $E_1,\dots,E_r$ are presented with dashed lines. The numbers denote the bounds on the $E_i$'s.}
\label{fig:tight-matroid}
\end{SCfigure}
\begin{proof}[Proof of Theorem~\ref{thm:laminar-intro}\ref{thm:laminar-upper}]
We will show a slightly stronger statement. 
Let $p\in (0,1)$ and define $a(p)=2-2p+\ln(p)$ if $p>1/e$, and $a(p)=(2-e)p$ if $p\le 1/e$.
Fix an $\eps>0$. By Claim~\ref{lem:limit}, there exists $r_\eps$ large enough such that 
\[
a(r_\eps, p) < a(p) + {\eps}/{2}.
\]
Fix a positive integer $q$ and consider the instance of laminar matroid secretary corresponding to $r_\eps$ as described above. Let $\cM=(E,\cI)$ denote the instance obtained by adding an infinite number of parallel copies at the end of the value order, so that $\cM$ satisfies Assumption~\ref{asm:copies}. By a slight abuse of notation, we still denote by $F_i$ the $i$-th set of elements described in the above construction together with their parallel copies, and the same applies to the $E_i$'s. 

Let $e^*$ denote the unique element in $\OPT(E) \cap E_1$, and consider the labeling scheme $\Lambda_{\text{chain}}(e^*)$  corresponding to a chain order $\pi$ with respect to $e^*$. Clearly, $e^*$ is accepted by Algorithm~\ref{alg:greedy-improving} if and only if $t_{e^*} > p$ and the algorithm has accepted at most $c-1$ element from $E_c$ in the time interval $({1}/{2}, t_{e^*})$ for each $c\in[r_\eps]$. By definition, for every pair of (original) elements $e \in F_i, e' \in F_j$ with $i < j$, we have $w_e > w_{e'}$. 

Let $T = \set{e \in E \colon t_e < p}$ denote the set of elements that arrived during the sampling period and $\calE$ be the event that $|\OPT(T) \cap F_i| = 1\text{ for all } i \in [r_\eps]$. Since $r_\eps$ is constant, we can choose $q$ large enough so that $\Pr[\calE] \ge 1-{\eps}/{2}$. Note that if $\calE$ holds then any improving element $e \in F_i$ arriving at time $t_e \in [p,t_{e^*})$ will be assigned a label $i$ upon arrival by the labeling scheme $\Lambda_{\text{chain}}(e^*)$.  

Let us now consider the improving word $z$ produced by the labeling scheme in $[p,1]$. Whenever $e^*$ arrives in $[p,1]$, it has label $1$ since it is the element of rank $1$ of $\OPT(t_{e^*}^+)$, and so $z$ has the form $x1y$ for some $x, y \in \{2,3,\ldots,r_{\varepsilon}\}^*$, where the symbol 1 refers to the entry of $z$ corresponding to $e^*$. Furthermore, the constraints described above imply that $e^*$ is accepted if and only if $y$ contains at most $c-1$ copies of label $c$ for each $c \in [r_{\eps}]$. It follows that under event $\calE$, $e^*$ is accepted if and only if $z$ is well-indexed. Since the probability that $z$ is well-indexed is $a(r_\eps,p)$ we get $\Pr[e^*\in \ALG_0] \leq \Pr[z\text{ is well-indexed} \vee \neg \calE] \leq a(r_\eps,p) + \Pr[\neg\calE] < a(p) + {\eps}/{2} + {\eps}/{2}$. We finish the proof by setting $p = {1}/{2}$ on the previous expression.
\end{proof}

\subsection{Rank-2 Matroids}\label{sec:rank2}

An interesting class of matroids for which the best achievable competitiveness is still unknown is the class of rank-2 matroids. A loopless rank-2 matroid $\calM=(E,\calI)$ has very simple structure since its circuits have only 2 or 3 elements. Recall that two elements $e$ and $f$ are called \emph{parallel} if $\{e,f\}$ is a circuit, and being parallel is an equivalence relation on $E$. Consider the equivalence classes $C_1,\dots, C_k$ of this relation. A set $F\in \calI$ is then independent if and only if $|F\cap C_i|\leq 1$ and $|F|\leq 2$. So, in particular, rank-2 matroids are laminar matroids and, by Lemma~\ref{lem:arp}, \textsc{Greedy-Improving} is at least $a(2,p)=p(2 - 2 p + p \ln(p))$-competitive. By setting $p=0.4241$, this leads to a $0.3341$-competitive algorithm for this class.

On the other hand, there are laminar matroids for which \textsc{Greedy-Improving} performs better. For example, on a uniform matroid of rank 2, the algorithm is $c(p,2)=p  (3 - 3 p + 2 p \ln(p)))=a(2,p)+p(1-p)$-competitive. 
Below, we analyze more carefully the competitiveness of \textsc{Greedy-Improving} for this class. Then we propose and analyze a second algorithm, called \textsc{Oblivious-Partition}, with a higher competitiveness on instances where \textsc{Greedy-Improving} behaves badly. 
Their combination leads to a new algorithm with better competitiveness for rank-2 matroids. 
We show the following result.

\begin{theorem}\label{thm:rank-2-intro}
There exists a $0.3462$-probability-competitive algorithm for the MSP on rank-2 matroids.
\end{theorem}

In what follows, we show how to prove the theorem.
Suppose that $\OPT(E)=\{e^*,f^*\}$. Let $z'$ be the improving word of $e^*$ generated by the labeling scheme $\Lambda_{\text{chain}}(e^*)$. As before, denote by $p\leq t_k<\dots <t_1\leq 1$ and $e_k,\dots, e_1$ the improving times and elements in $[p,1]$, respectively, and suppose that $e^*=e_i$. In particular, we know that $z'_1=z'_2=\dots=z'_{i-1}=2$ and $z'_i=1$. Note that $e^*$ is selected as long as $e^*$ is not spanned by $A=\{e_{i+1},\dots,e_k\}$. If the matroid is uniform, it is sufficient that $|A|\leq 1$ to achieve this. For general rank-2 matroids though, $e^*$ is selected if $|A|=0$, if $|A|=1$ and the unique element $e_k$ of $A$ is not parallel to $e^*$, or if $|A|\geq 2$ and all elements of $A$ are parallel to each other but not parallel to $e^*$. Let $\GRE$ be the output of \textsc{Greedy-Improving}. The previous discussion implies that
\[
\Pr[e^*\in \GRE] \geq \Pr[z'\in 2^*1]+\Pr[z'\in 2^*12] + \Pr[z'\in 2^*11 \wedge \text{$e_{|z'|}$ is not parallel to $e^*$}].
\]
Recall that $a(2,p)=\Pr[z'\text{ is well-indexed}] = \Pr[z'\in 2^*1]+\Pr[z'\in 2^*12]$, so we want to understand in which situations the extra term above is strictly positive.

In what follows, we will use an alternative analysis in which we condition on the arrival time $t$ of element $e^*$. Let $\{g,h\}=\OPT(E_t)$ be the optimal set immediately before the arrival of $e^*$, and $\OPT(E_t^+)=\OPT(E_t\cup \{e^*\})=\{g,e^*\}$. In particular, $h$ is the element of $\OPT(E_t)$ that $e^*$ replaces on arrival, and $g$ is not parallel to $e^*$. We define a new labeling scheme $\Lambda_\tau$ only for the interval $[p,t)$. The order $\tau$ is chosen to be any total order of $E_t$ having $g$ as its first element. In particular, if $p<s<t$ and there is no improving time in $[s,t)$, then $\OPT(E_s) = \OPT(E_t) = \{g,h\}$ and the labels of $g$ and $h$ at time $s$ are 1 and 2, respectively.

Let $z$  denote the improving word on $[p,t)$ with respect to $\Lambda_\tau$. Note that if $z=\epsilon$ is the empty word, then $e^*$ is selected by \textsc{Greedy-Improving}. If $z=1$, then $g$ is the only improving element that arrived before $e^*$, so $\GRE=\{e^*,g\}$. If $z=2$, then $e^*$ is selected by the algorithm if and only if $h$ is not parallel to $e^*$. The algorithm might select $e^*$ for longer improving words as well, but we only focus on these three scenarios. 

For each $s\in [0,1]$, let $\Para(s)$ denote  the event that $\OPT(E_s^+ - \{e^*\})$ contains an element parallel to $e^*$, and let $\alpha(s)=\Pr[\Para(s)]$. Based on the facts that the arrival time of $e^*$ is uniform in $[0,1]$, that the length of $z$ behaves like a Poisson random variable with rate $\lambda = 2 \ln\prn{{t}/{p}}$, and that each symbol of $z$ is chosen uniform from the set $\{1,2\}$, the above discussion implies that
\begin{align}
\Pr[e^*\in \GRE]&=  \int_{p}^1 \Pr[e^*\in \GRE | t_{e^*}=t ]\diff t\nonumber\\
&\geq \int_p^1 \left(\Pr[z=\epsilon]  + \Pr[z=1] + \Pr[z=2| \neg\Para(t)]\Pr[\neg\Para(t)] \right) \diff t\nonumber\\
&= \int_p^1 \left(e^{-\lambda} + \frac{1}{2}\frac{e^{-\lambda}\lambda}{1!}+ 
 \frac{1}{2}\frac{e^{-\lambda}\lambda^1}{1!}(1-\alpha(t)) \right)\diff t\nonumber\\
&= \int_p^1 \prn{ \prn{\frac{p}{t}}^2 + \prn{\frac{p}{t}}^2 \ln\prn{\frac{t}{p}} + \prn{\frac{p}{t}}^2 \ln\prn{\frac{t}{p}} (1-\alpha(t)) } \diff t \label{eq:greedyrank2}\\
&= p (2 - 2 p + p\ln(p)) + \int_p^1 \prn{\frac{p}{t}}^2 \ln\prn{\frac{t}{p}} (1-\alpha(t)) \diff t, \nonumber
\end{align}
where the third line follows since for fixed $t$, the event $z = 2$ has the same probability independent from $\neg \Para(t)$, so $\Pr\brk{z=2 \midd \neg\Para(t)} = \Pr[z=2]$.

For certain rank-2 matroids and value orders, $(1-\alpha(t))$ is equal to 1 as a function of $t$ on the interval $[p,1]$ (e.g., for rank-2 uniform matroids), while for others the function is close to 0. Motivated by this observation, we propose a different algorithm that has higher probability of accepting $e^*$ whenever $(1-\alpha(t))$ is small, i.e., when it is likely that $\OPT(E_{t})$ contains an element parallel to $e^*$. At the same time, our goal is not to loose too much on the case when $\OPT(E_t)$ does not contain an element parallel to $e^*$. The algorithm, called \textsc{Oblivious-Partition}, is presented as Algorithm~\ref{alg:oblivious-rank2}.

\begin{algorithm}[H]
    \begin{algorithmic}[1]
    \State Initialize $\ALG_1 \gets \emptyset$.
    \State Skip all elements with arrival time in $[0,p)$.
    \State Let $\OPT(E_p)=\{g_1,g_2\}$, with $g_1\succ g_2$.
    Let $C_1$ be the parallel class of $g_1$.
    \State Include into $\ALG_1$ the first improving element $f_1$ arriving in $[p,1]$ such that $f_1$ is parallel to $g_1$, and the first improving element $f_2$ arriving in $[p,1]$ such that $f_2$ is not parallel to $g_1$.
    \State Return $\ALG_1$.
    \end{algorithmic}
    \caption{\textsc{Oblivious-Partition Algorithm  for Rank-2 Matroids}}
    \label{alg:oblivious-rank2}
\end{algorithm}
Let $\PAR$ be the output of the algorithm. For the analysis, we choose $e^*\in \OPT(E)$ with arriving time $t=t_{e^*}\in [p,1]$, denote by $\OPT(E_t^+)=\{g,e^*\}$, and consider as before a labeling scheme $\Lambda_\tau$ where $\tau$ is any total order of $E_t$ having $g$ as its first element. Similarly as before, let $z$ denote the improving word on $[p,t]$ with respect to $\Lambda_\tau$. We distinguish two cases.\\

\noindent{\bf Case 1: $\Para(t)$ holds.} In this case, $\OPT(E_t)=\{g,h\}$ is such that $h$ is parallel to $e^*$. If $h$ arrives before time $p$ then it is also part of $\OPT(E_p)$, and the word $z$ does not contain a symbol 1. In fact, $h$ is contained in $\OPT(E_s)$ for all $s\in [p,t]$, therefore no other improving element parallel to $h$ will appear in $[p,t)$. Here we have two sub-cases depending on whether $h$ is $g_1$ or $g_2$.

If $h=g_1$, then $e^*$ is selected by the algorithm as the first improving element parallel to $g_1$ that arrives. If $h = g_2$, then $g_1\succ h$. This means that $\OPT(E_t)$ consists of the elements $h$ and $g$, where $g\succ h$ in the value order. In fact, since the matroid has rank 2, every element $f$ in $E_t$ that satisfies $f\succ h$ must be parallel to $g$. In particular, $g_1$ is parallel to $g$. Observe also that if $f'$ is an improving element arriving at time $s\in [p,t)$, then $\OPT(E_s^+)=\{f',h\}$. Since  $\{g_1,h\}$  is an independent subset of $E_{s}^+$, we conclude that $f'\succ g_1\succ h$, which implies that $f'$ is parallel to $g$, and so it is parallel to $g_1$ as well. This means that $e^*$ is the first improving element arriving that is not parallel to $g_1$, thus it is selected by the algorithm.

Summarizing, if $\Para(t)$ holds and $h$ arrives before time $p$, then $e^*$ is selected. Therefore, the events $\Para(t)$ and $z\in 2^*$ together imply that $e^*$ is selected.\\

\noindent{\bf Case 2: $\Para(t)$ does not hold.} 
In this case, $\OPT(E_t)=\{g,h\}$, $\OPT(E_t^+)=\{g,e^*\}$ and $g$, $h$ and $e^*$ are all from different classes. Note that if no improving elements arrive in $[p,t)$, then $e^*$ is selected by the algorithm as $f_2$. This means that the events $\neg \Para(t)$ and $z=\epsilon$ together imply that $e^*$ is selected.

To conclude, note that $z\in 2^*$ if and only if its restriction $z_1$ to alphabet $1$ is empty. Recall that $|z_1|$ and $|z|$ both follow Poisson distributions with parameters $\lambda_1=\ln\prn{{t}/{p}}$ and $\lambda=2 \ln\prn{{t}/{p}}$, respectively. We conclude that
\begin{align}
\Pr[e^*\in \PAR]&=  \int_{p}^1 \Pr[e^*\in \PAR \mid t_{e^*}=t ]\diff t\nonumber \\
&\geq \int_p^1 \prn{\Pr\brk{z_1=\epsilon \mid \Para(t)} \Pr[\Para(t)]  + \Pr\brk{z=\epsilon \mid \neg\Para(t)} \Pr[\neg\Para(t)] } \diff t \nonumber  \\
& = \int_p^1 \prn{ e^{-\lambda_1} \alpha(t) + e^{-\lambda} (1 - \alpha(t)) } \diff t \nonumber \\
& = \int_p^1 \prn{ \prn{\frac{p}{t}} \alpha(t) + \prn{\frac{p}{t}}^2 (1-\alpha(t)) } \diff t \label{eq:Partrank2},
\end{align}
where the third line follows since the words $z_1$ and $z$ are independent from  $\Para(t)$.\\

\noindent{\bf The Mixture Algorithm.} The above analysis shows that \textsc{Greedy-Improving} performs well when $\Pr[\neg\Para(t)]$ is large, while \textsc{Oblivious-Partition} performs well when $\Pr[\Para(t)]$ is large. The idea is to combine these two algorithms to get a better probability-competitive guarantee. Therefore, our final algorithm for rank-2 matroids is a mixture of the two. The algorithm, called \textsc{Mixture}, is presented as Algorithm~\ref{alg:mixture-rank2}.

\begin{algorithm}[H]
    \begin{algorithmic}[1]
\State Fix the sample size $p \in (0,1)$.

\State With probability $\eps$, run \textsc{Oblivious-Partition} and return $\ALG_1$.

\State Otherwise, run \textsc{Greedy-Improving} and return $\ALG_0$.
    \end{algorithmic}
    \caption{\textsc{Mixture Algorithm for Rank-2 Matroids}}
    \label{alg:mixture-rank2}
\end{algorithm}

\begin{proof}[Proof of Theorem~\ref{thm:rank-2-intro}]
We choose $e^*\in \OPT(E)$ with arriving time $t=t_{e^*}\in [p,1]$. Let $\ALG$ denote the output of Algorithm~\ref{alg:mixture-rank2}. Then, by \eqref{eq:greedyrank2} and \eqref{eq:Partrank2}, we conclude that 
\begin{align*}
    &\Pr[e^*\in \ALG] \\ 
    &\geq (1-\eps) p (2-2p+p\ln(p)) + \int_p^1\alpha(t)\cdot \eps \frac{p}{t} \diff t\\
    &\quad + \int_p^1 (1-\alpha(t)) \prn{ (1-\eps) \prn{\frac{p}{t}}^2\ln\prn{\frac{t}{p}} + \eps \prn{\frac{p}{t}}^2 } \diff t \\
    &\geq 
    (1-\eps) p (2-2p+p\ln(p))+ \int_p^1 \min\prn{ \eps \frac{p}{t}, (1-\eps) \prn{\prn{\frac{p}{t}}^2 \ln\prn{\frac{t}{p}} + \eps \prn{\frac{p}{t}}^2 } } \diff t.
    \end{align*}
Fix $p\in [0,1]$ and consider the expressions inside the minimum, that is,
\begin{align*}
G_1(\eps,t) &= \eps \frac{p}{t}, \\
G_2(\eps,t) &= (1-\eps) \prn{\frac{p}{t}}^2 \ln\prn{\frac{t}{p}} + \eps \prn{\frac{p}{t}}^2= \prn{\frac{p}{t}}^2 \ln\prn{\frac{t}{p}} +\eps \prn{1 - \ln\prn{\frac{t}{p}} } \prn{\frac{p}{t}}^2.
\end{align*}

Note that for every given value $q > p$, both $G_1(\eps,q)$ and $G_2(\eps,q)$ are linear functions in $\eps$.
Furthermore, we have $G_1(0,q) = 0 < G_2(0,q) = \prn{{p}/{q}}^2 \ln\prn{{q}/{p}}$. Therefore, $G_1(\eps,q) \leq G_2(\eps,q)$ if and only if $\eps \leq {p\ln\prn{{q}/{p}}}/{(q - p + p \ln\prn{{q}/{p}})}$.
Let $q_\eps > p$ be the unique solution of $\eps \prn{q_\eps -p + p \ln\prn{{q_{\eps}}/{p}}} = p \ln\prn{{q_{\eps}}/{p}}$. Then, 
\[
\min\prn{G_1(\eps,t), G_2(\eps,t)} = 
\begin{cases} 
G_1(\eps,t) &\text{ if $p < t < q_\eps$,} \\
G_2(\eps,t) &\text{ if $t \geq q_\eps$.}
\end{cases}
\]
In particular, the lower bound for $\Pr[e^*\in\ALG]$ can be rewritten as
\begin{eqnarray*}
(1-\eps) p (2-2p+p\ln(p)) + \int_p^{\min(q_\eps,1)} \eps \frac{p}{t} \diff t 
+ \int_{\min(q_\eps,1)}^1 \prn{ (1-\eps)\prn{\frac{p}{t}}^2 \ln\prn{\frac{t}{p}} + \eps \prn{\frac{p}{t}}^2 } \diff t.
\end{eqnarray*}
The value of this expression is equal to
\[
\begin{cases}
(1-\eps) p(2-2p+p\ln p) + \eps (p \ln(p)) & \text{if } q_\eps \geq 1 \vspace{1em} \\

\begin{aligned}
&(1-\eps) \prn{ p(2-2p+p\ln(p)) + \frac{p^2}{q_\eps} (1-q_\eps)(1-\ln(p)) + \frac{p^2}{q_\eps} \ln(q_\eps) } \\
& \quad + \eps \prn{ \frac{p^2}{q_\eps} (1-q_\eps) + p \ln\prn{\frac{q_\eps}{p}} }
\end{aligned}
& \text{if } q_\eps\in (p,1)
\end{cases}
\]
Note that whenever $q_\eps > 1$, then we recover a ratio that is at most $p(2-2p+p\ln p)$, which is a lower bound for the ratio of \textsc{Greedy-Improving}, so the best ratio is achieved for $q_\eps\leq 1$. To find the maximum, we set 
$\eps = {p \ln\prn{{q}/{p}}}/{(-p+q+\ln\prn{{q}/{p}})}$ and obtain an explicit maximization  problem in two variables $p \leq q \leq 1$. The maximum can be computed numerically to get $p\approx 0.4067$, $q\approx 0.9194$, $\eps\approx 0.3928$, showing a probability-competitive ratio of $0.3462$ for \textsc{Mixture}.
\end{proof}
%%%%%%%%%%%%%%%%
%%%%%%%%%%%%%%%%
\section{Graphic Matroids}
\label{sec:graphic}
%%%%%%%%%%%%%%%%

Given a graph with edges $E$ as the ground set, the independent sets in the {\it graphic matroid} are edge sets without cycles, i.e., $\calI=\{F\subseteq E\colon F\ \text{is acyclic}\}$. In this section, we use our labeling schemes framework to improve the state-of-the-art competitive ratio for the graphic matroid secretary. Our main result is the following.
\begin{theorem}\label{thm:graphic-intro}
For the MSP in graphic matroids, the following holds:
\begin{enumerate}[label=\normalfont(\roman*)]
    \item There is a $0.2693$-probability-competitive algorithm for simple graphs.\label{thm:graphic-simple}
    \item There is a $0.2504$-probability-competitive algorithm for general graphs.\label{thm:graphic-general}
\end{enumerate}
\end{theorem}

To prove the theorem, in Section~\ref{sec:basic} we first describe a basic algorithm that uses a simple labeling scheme and can be considered a variant of the 2-forbidden algorithm of~\cite{forbidden-paper}. In Section~\ref{sec:generation}, we introduce one of the key notions of this section, the {\it generation} of an edge, which suggests a new algorithm that we analyze through our labeling scheme framework to show the improved $0.2693$-probability-competitiveness for simple graphs, i.e., with no parallel edges. To tackle the case of general graphs, we combine this algorithm with another that performs well on different types of instances, breaking the $1/4$ barrier. 
%All missing proofs can be found in the full version of this paper.

%%%%%%%%%%%%%%%%
\subsection{Basic Algorithm and Auxiliary Digraph}
\label{sec:basic}
%%%%%%%%%%%%%%%%

Let $G'=(V',E')$ denote the input graph of the graphic MSP. For graphic matroids, we can implement Assumption~\ref{asm:copies} with the following construction.

\begin{assumption} \label{asm:graphic} 
The input graph of our algorithm is the augmented graph $G=(V+w, E)$ with $E=E'\cup F$, where $w$ is a new node considered a {\it root} and $F$ is a collection of dummy edges consisting of an infinite number of parallel copies of the edges $\{wv\colon v\in V\}$. These dummy edges have infinitesimal weight. 
\end{assumption}

Recall that every edge $e$ in $E=E'\cup F$ picks its arrival time $t_e$ from the uniform distribution over $[0,1]$. Under Assumption~\ref{asm:graphic}, for every $t>0$, the set $E_t$ of all edges in $E=E'\cup F$ arriving in the interval $[0,t)$ has full rank $r=(|V+w|-1)=|V|$, and thus $\OPT(E_t)$ is a tree with vertex set $V+w$. 

\begin{definition}[Canonical Orientation of an Improving Edge]
The arborescence $A(t)$, resp. $A(t^+)$, obtained by orienting the tree $\OPT(E_t)$, resp. $\OPT(E_t^+)$, away from the root $w$ is called the canonical orientation of the optimum at $t$, resp. $t^+$. For an edge $e\in E$, we say that $e$ is oriented from $u$ to $v$ at time $t$ if $e\in \OPT(E_t^+)$ and $e$ is oriented as $(u,v)$ in $A(t^+)$. 
\end{definition}

Note that the orientation of an edge changes from improving time to improving time and that only improving elements are oriented. Indeed, if $e=\{u,v\}$ is an edge arriving at time $t_e$, then, at that time, it receives an orientation $(u,v)$ according to the arborescence $A(t_e^+)$. When the next improving edge arrives, say $f$, the optimum tree $\OPT(E_{t_f}^+)$ is different from $\OPT(E_{t_e}^+)$ since $f$ enters and some other edge leaves the tree. If $e$ is still part of the optimum tree, then its orientation may have changed; see Fig.~\ref{fig:arb} for an illustration.
\begin{figure}[t]
    \centering
    \includegraphics[width=\textwidth]{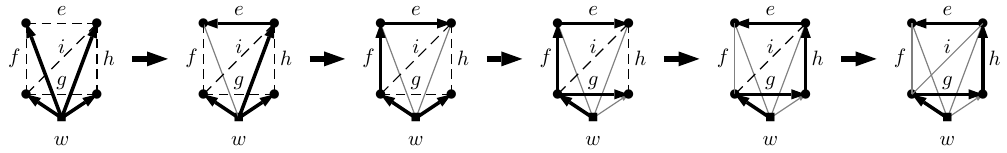}
    \caption{An example for the canonical orientation of improving edges. Vertex $w$ is chosen as the root, the original edges of the canonical orientation are drawn with thick lines, and a set $\{e,f,g,h,i\}$ arrives in the order $e,f,g,h,i$, with ranking $e > g > h > i > f$. Note that $e,f,g$ and $h$ are improving edges whereas $i$ is not, the orientation of $e$ changes several times, and $f$ leaves the optimal tree upon the arrival of $h$.}
    \label{fig:arb}
\end{figure}
The following algorithm is a variant of the 2-forbidden algorithm of~\cite{forbidden-paper} and is the foundation of our algorithms for graphic matroids. 
%Algorithm~\ref{alg:aux-graphic} is $p(1-p)$-probability-competitive; see the full version for details. In particular, $1/4$-probability-competitive for~$p = 1/2$.
\begin{algorithm}[H]
    \begin{algorithmic}[1]
    \State Initialize $\ALG_0 \gets \emptyset$.
    \State Initialize $\AUX \gets$ digraph $(V+w,\emptyset)$.
    \State Skip all edges with arrival time in $[0,p)$.
   \For{each edge $e$ with $t_e\in [p,1]$, in arrival order}
        \If{$e \in \OPT(E_{t_e}^+)$ is oriented from $u$ to $v$ at time $t_e$ and $\deg^-_{\AUX}(v)=0$}
            \State $\AUX \gets \AUX + (u,v)$
            \If{$\deg^-_{\AUX}(u)=0$}
                \State $\ALG_0 \gets \ALG_0 + e$
            \EndIf    
        \EndIf
    \EndFor    
    \State Return $\ALG_0$.
    \end{algorithmic}
    \caption{\textsc{Basic Algorithm}}
    \label{alg:aux-graphic}
\end{algorithm}
\vspace{.5cm}
\begin{theorem}\label{thm:basiccompetitive} 
Algorithm~\ref{alg:aux-graphic} is $p(1-p)$-probability-competitive. 
In particular, it is $1/4$-probability-competitive when $p = 1/2$.
\end{theorem}

To show the above theorem, we first show, through the following lemmas, that Algorithm~\ref{alg:aux-graphic} returns a feasible solution.

\begin{lemma}
At every step, the digraph $\AUX$ has maximum in-degree at most one. Therefore, every cycle of the underlying undirected graph has to be a directed cycle in $\AUX$, and $\AUX$ has at most one directed cycle per connected component. 
\end{lemma}
\begin{proof}
    An edge $e$ oriented from $u$ to $v$ is added to $\AUX$ only if the in-degree of $v$ is 0, showing the first half of the statement. The second half is easy to see and is a folklore result in graph theory.  
\end{proof}

\begin{lemma}
$\ALG_0$ is an independent set in the graphic matroid.
\end{lemma}
\begin{proof}
Since $\ALG_0$ is a subset of the underlying graph of $\AUX$, the only possibility for a cycle $C$ to be present in $\ALG_0$ is that $C$ corresponds to some directed cycle in $\AUX$. 
    However, the last arriving edge $e$, with orientation  $(u,v)$ in $\AUX$, cannot enter $\ALG_0$ since its tail $u$ has, at that moment, in-degree one in $\AUX$. 
    Therefore, there are no cycles in $\ALG_0$.   
\end{proof}
\noindent{\bf Labeling Scheme $\Lambda_0$.} We now present the labeling scheme we use to analyze Algorithm~\ref{alg:aux-graphic}.
Let $e^*\in \OPT(E)$ be an optimal edge. In particular, $e^*\in E'$ is not a dummy edge. Let $t^*=t_{e^*}$ be its arrival time and $\vec{e}=(u^*,v^*)$ be the orientation of $e^*$ at time $t^*$. We are ready to define the labeling scheme that we use to analyze the basic algorithm.

%\begin{definition}[Labeling Scheme $\Lambda_0(e^*)$]
On the interval $[t^*,1]$, we use the usual labeling scheme $\Lambda_\pi$ induced by any total order of $E$ such that $\pi_1=e^*$. For times $s<t_{e^*}$, we use a different labeling that depends on $u^*$ and $v^*$.  Recall that for every vertex $v\in V$, there is exactly one edge oriented toward $v$ at time $s$, namely the unique edge in $\OPT(E_s)$ whose associated arc in $A(s)$ is oriented toward $v$. We assign label 1 to the edge oriented toward the head $v^*$ of $e^*$, and label 2 to the edge oriented toward the tail $u^*$. The remaining $r-2$ labels are assigned arbitrarily to the other elements of $\OPT(E_s)$.
%\end{definition} 

Denote by $\cL_{\text{basic}}=\{z\in [r]^*\colon z=x1y$, where $x$ is a word not containing the symbol 1 and $y$ is a word not containing the symbols 1 and 2$\}$. The next lemma connects Algorithm~\ref{alg:aux-graphic} with this language.

\begin{lemma}\label{lem:wordbasic}
Let $z$ be the improving word of $e^*\in\OPT(E)$ produced by the labeling scheme $\Lambda_{0}(e^*)$. If $z\in \cL_{\normalfont\text{basic}}$ then $e^*\in \ALG_0$.
\end{lemma}
\begin{proof}
By Lemma~\ref{lem:standardlabel}, the arrival time $t^*$  of $e^*$ is in $[p,1]$ if and only if there exists a symbol 1 in $z$. Furthermore, all the symbols after the first 1 are labels of improving edges arriving in $[p,t^*)$. Observe that $\vec{e}\in \AUX$ if and only if $t^*\in [p,1]$ and no improving edge arriving in $[p,t^*)$ is oriented toward $v^*$ upon arrival, which corresponds to the condition that there is a unique symbol 1 in $z$ by the definition of the dynamic labeling we are using.

Observe that $e^*\in \ALG_0$ if and only if $\vec{e}\in \AUX$ and no improving edge arriving in $[p,t^*)$ is oriented toward $u^*$ upon arrival. By the previous paragraph and the labeling we use, this is equivalent to $z=x1y$, where $x$ is a word without symbols in $\{1\}$ and $y$ is a word without symbols in $\{1,2\}$.
\end{proof}

We are now ready to show Theorem~\ref{thm:basiccompetitive}.

\begin{proof}[Proof of Theorem~\ref{thm:basiccompetitive}]
Let $e^*\in \OPT(E)$ and $z$ be the improving word produced by the associated labeling scheme $\Lambda_0(e^*)$
in the interval $[p,1]$. Instead of looking at the produced word $z\in [r]^*$ directly, we look at the subsequences $z_{\{1,2\}}$, $z_{\{1\}}$ and $z_{\{2\}}$ obtained by restricting $z$ to alphabets $\{1,2\}$, $\{1\}$ and $\{2\}$, respectively. 

By Lemma~\ref{lem:words}, $z_{\{1\}}\in 1^*$ and $z_{\{2\}}\in 2^*$ are random words whose length is distributed as a Poisson with parameter $\lambda_1=\log(1/p)$. By Lemma~\ref{lem:wordbasic}, we get
\begin{align*}
\Pr[e^*\in \ALG_0]&\geq \Pr[z\in ([r]\setminus \{1\})^* 1 ([r]\setminus \{1,2\})] \\
&= \Pr[z_{\{1,2\}}\in 2^*1]\\
& = \Pr[z_{\{1,2\}}\in 1^*1] \\
& = \Pr[z_2=\epsilon]\cdot \Pr[z_1\neq \epsilon]= e^{-\lambda_1}(1-e^{-\lambda_1})=p(1-p),
\end{align*}
where the third line holds since generating the word $2^k1$ has the same probability than generating the word $1^k1$ for all $k\in \NN$, while the fourth line holds by the independence of words $z_1$ and $z_2$ given by Lemma~\ref{lem:words}.
\end{proof}

%%%%%%%%%%%%%%%%
\subsection{Cycles and Generations of the Auxiliary Digraph}
\label{sec:generation}
%%%%%%%%%%%%%%%%

Let $G'=(V,E')$ be an undirected, not necessarily simple graph, and $G=(V+w, E)$ with $E=E'\cup F$ be defined as before. In the following, we study the cycles of the auxiliary digraph $\AUX$ in Algorithm~\ref{alg:aux-graphic} when applied to the augmented graph $G$. 
For that, enumerate all arcs in $\AUX$ in the order in which they entered as $(u_1,v_1), \dots (u_k,v_k)$. Note that $\{v_1,\dots, v_k\}\subseteq V$ and they are all different, while $\{u_1,\dots, u_k\}\subseteq V+w$ and they might repeat.
For every time $t\in [0,1]$, let $\AUX(t)$ be the sub-digraph of $\AUX$ consisting of all arcs that entered $\ALG_0$ in the interval $[0,t]$. 

\begin{definition}[Generation of an arc in $\AUX$]
Let $e$ be an edge such that its orientation $\vec{e}=(u,v)$ entered $\AUX$ at time $t$. The \emph{generation} $\Gen(e)\in \N$ is defined as follows. If $e\in F$ is a dummy edge, we set $\Gen(e)=1$. If $e\in E'$ and $\deg^-_{\AUX(t)}(u)=0$, we set $\Gen(e)=0$. Otherwise, there is a unique arc $\vec{f}$ with head $u$ that entered $\AUX$ at a time $s<t$, and we set $\Gen(e)=\Gen(f)+1$.
\end{definition}

Let $C\subseteq E'\cup F$ be a cycle in $\AUX$. No edge in $C$ can be in $F$, since the edges of $F$ are oriented away from $w$ by construction, and no arc in $\AUX$ is directed toward $w$. 
Fig.~\ref{fig:cycle} shows an example of a cycle and the generations of its edges.
Also, the edges $e$ that enter $\ALG_0$ are exactly those with $\Gen(e)=0$.
\begin{SCfigure}[50][t]
\centering
\begin{minipage}[b]{0.35\textwidth}
    \centering
        \begin{tikzpicture}[scale=.95] % Adjust the scale as needed
        % Define the coordinates for the 5-cycle
        \node (A) at (90:2) {\scriptsize$w_1$};
        \node (B) at (18:2) {\scriptsize$w_3$};
        \node (C) at (306:2) {\scriptsize$w_4$};
        \node (D) at (234:2) {\scriptsize$w_2$};
        \node (E) at (162:2) {\scriptsize$w_5$};

        % Draw the directed edges
        \draw[->] (A) -- (B) node[midway, above right] {\scriptsize$f_3, 1$};
        \draw[->] (B) -- (C) node[midway, below right] {\scriptsize$f_4, 2$};
        \draw[->] (C) -- (D) node[midway, below] {\scriptsize$f_2, 0$};
        \draw[->] (D) -- (E) node[midway, below left] {\scriptsize$f_5,1$};
        \draw[->] (E) -- (A) node[midway, above left] {\scriptsize$f_1,0$};
    \end{tikzpicture}
\end{minipage}%
\caption{A directed 5-cycle in $\AUX$ arriving in the order $f_1, f_2, f_3, f_4, f_5$; each arc is marked as $f_i, \Gen(f_i)$.}
\label{fig:cycle}
\end{SCfigure}
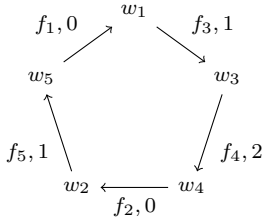
\begin{lemma}\label{lem:generation1} 
Let $\AUX'$ be the sub-digraph of $\AUX$ obtained by deleting all edges with generation equal exactly to one. Then $\AUX'$ is acyclic.
\end{lemma}
\begin{proof}
   Suppose to the contrary that $\AUX'$ contains a directed cycle $C$ and let $f=(u_1,u_2)$ be the first arriving arc in $C$. Let $h=(u_k,u_1)$ be the previous arc of $C$ in cyclic order. Since $h$ is added to $\AUX$, no arc entered its head $u_1$ at its arrival. Consequently,  no arc entered the tail of $f$ upon its arrival. Hence, the generation of $f$ is $0$. This implies that the generation of the next arc $g=(u_2,u_3)$ in cyclic order is $1$, contradicting the assumption that $\AUX'$ has no edges with generation equal to $1$. Note that this proof works even if $k=2$, that is when $h=g$.
\end{proof}
The previous discussion and the fact that the generation of each arc can be computed easily at its inclusion to $\AUX$ suggest the following algorithm.
\begin{algorithm}[h]
    \begin{algorithmic}[1]
    \State $\ALG_1 \gets \emptyset$
    \State Initialize $\AUX \gets$ digraph $(V+w,\emptyset)$.
    \State Skip all edges with arrival time in $[0,p)$.
   \For{each edge $e$ with $t_e\in [p,1]$, in arrival order}
        \If{$e \in \OPT(E_{t_e}^+)$ is oriented from $u$ to $v$ at time $t_e$ and $\deg^-_{\AUX}(v)=0$}
            \State $\AUX \gets \AUX + (u,v)$
            \If{$u\neq w$}
                \If{$\deg^-_{\AUX}(u)=0$}
                \State $\ALG_1 \gets \ALG_1 +   e$
                \Else{ let $f=(u',u)$ be the unique arc in $\AUX$ pointing toward $u$}
                \If{$u'=w$ or $\deg^-_{\AUX(t_f)}(u')=1$}
                    \State $\ALG_1 \gets \ALG_1 +e$
                \EndIf    
            \EndIf    
        \EndIf
        \EndIf
    \EndFor    
    \State Return $\ALG_1$.
    \end{algorithmic}
    \caption{\textsc{Generation Algorithm}}
    \label{alg:generation}
\end{algorithm}
\begin{lemma}\label{lem:feasiblegen}
Algorithm~\ref{alg:generation} outputs a set $\ALG_1$ that is independent.
\end{lemma}
\begin{proof}
Note that dummy edges $e\in F$ do not enter $\ALG_1$. For a non-dummy edge $e$, the condition $\deg^-_{\AUX}(u)=0$ is satisfied if and only if $\Gen(e)=0$. The condition $u'=w$ is satisfied if and only if $f$ is a dummy edge, $\Gen(f)=1$, and  $\Gen(e)=2$. Finally, the condition $u'\neq w$ and  $\deg^-_{\AUX(t_f)}(u)=1$ holds if and only if $f$ is not a dummy edge, $\Gen(f)\geq 1$ and $\Gen(e)=\Gen(f)+1\geq 2$. It follows that $\ALG_1$  contains exactly the edges with generation different from 1. By Lemma~\ref{lem:generation1}, $\ALG_1$ is acyclic.
\end{proof}
We show that if $G$ does not contain parallel edges, the competitive ratio of \textsc{Generation} is strictly better than $1/4$. Let $G'=(V,E')$ be a simple graph. Since the case $|V|=2$ is straightforward, we assume that $|V|\geq 3$. Let $\AUX$ and $\ALG_1$ be obtained by applying Algorithm~\ref{alg:generation} to $G=(V+w,E'\cup F)$, and let $e^*\in \OPT(E)$ be an optimal edge arriving at a time  $t^*=t_{e^*}$. \\

\noindent{\bf Labeling Scheme $\Lambda_1$.}
We now introduce the labeling scheme we use to analyze the \textsc{Generation Algorithm}.
On the time interval $[t^*,1]$, we use the usual labeling scheme $\Lambda_\pi$ induced by any total order of $E$ such that $\pi_1=e^*$. Suppose that $e^*$ is oriented from a vertex $u^*$ to a vertex $v^*$ at time $t^*$, and let $w^*_0$ be an arbitrary vertex in $V\setminus \{u^*,v^*\}$. For $s<t^*$, the labeling scheme is as follows: let $\vec{f}=(u',u^*)$ be the earliest improving arc in $(s,1]$ having head $u^*$. If $\vec{f}$ exists and $u'\neq w$, then we set $w^*_s= u'$; since the original graph $G'=(V,E')$ is simple, $f$ is not parallel to $e^*$, showing that $u'\neq v^*$ holds as well. Otherwise, we set $w^*_s=w^*_0$. In any case, the vertices $v^*$, $u^*$, and $w^*_s$ are different. Recall that for any vertex $v\in V$, there is a unique edge in $\OPT(E_s)$ that is oriented toward $v$. Then, at time $s$, we assign label 1 to the edge oriented toward $v^*$, label 2 to the edge oriented toward $u^*$, and label 3 to the edge oriented toward $w^*_s$. The remaining $(r-3)$ labels are assigned arbitrarily to the other elements of $\OPT(E_s)$.     

Denote by $\cL_{\text{generation}}=\{z\in [r]^*\colon z_{\{1,2,3\}}=x1y$, where $x$ and $y$ are words not containing the symbol 1 and $y$ is a word that does not end with symbol 2$\}$. Here, $z_{\{1,2,3\}}$ denotes the subsequence of $z$ restricted to the alphabet $\{1,2,3\}$.
The following lemma connects Algorithm~\ref{alg:generation} with this language.
\begin{lemma}\label{lem:graphic-generation}
Let $z$ be the improving word of $e^*\in\OPT(E)$ produced by $\Lambda_{1}(e^*)$. If $z\in \cL_{\normalfont\text{generation}}$ then $e^*\in \AUX$ and $\Gen(e^*)\neq 1$. In particular, $e^*\in \ALG_1$.
\end{lemma}
\begin{proof}
Let $z'=z_{\{1,2,3\}}$ and suppose that $z\in \cL_{\text{generation}}$. That is $z'=x1y$; $x,y\in \{2,3\}^*$ and $y$ does not end with symbol 2. Let $\vec{e}=(u^*,v^*)$ with $u^*\neq w$ be the orientation of $e^*$ upon arrival and $t^*\in (0,1)$ be its arrival time. Since $z'$ has a unique 1, the same holds for $z$, implying that $t^*\in[p,1]$ and that $e^*$ is the unique improving edge that is oriented toward $v^*$ in $[p,t^*]$. Therefore, $\vec{e}\in \AUX$. We distinguish two cases based on whether $y$ contains a symbol $2$ or not.

If $y$ contains no symbol 2 at all, then no improving edge arriving in $[p,t^*)$ is oriented toward $u^*$ or $v^*$ in $[p,t^*)$. Since $e^*$ is not a dummy edge, we get $\Gen(e^*)=0$ and thus $\Gen(e^*)\neq 1$.

If $y$ contains a symbol $2$, then let $f$ be the first improving edge with arrival time $p\leq s<t$ that is oriented toward $u^*$, that is, $f=(u',u^*)$ for some $u'\in (V+w)-u^*$. Note that $f$ is the edge whose label is recorded as the last 2 of $y$. Furthermore, $\Gen(e^*)=\Gen(f)+1$. If $u'=w$ then, by definition, $f$ is a dummy edge in $F$, $\Gen(f)=1$ and $\Gen(e^*)=2$. So assume that $u'\neq w$, which means that $f\in E'$. Since the original graph $G'=(V,E')$ is simple, $f$ is not parallel to $e^*$, implying $u'\neq v^*$. It follows that, for every $s'$ with $p\leq s'<s$, we have $w^*_{s'}=u'$. Let $g$ be the first improving edge in $[p,s)$ that is oriented toward $u'$. Then its orientation  $\vec{g}=(u'',u')$ is in $\AUX$ and, by the previous discussion, $g$ is the edge whose label is recorded as the last 3 of the word $y$. This means that  $\Gen(e^*)=\Gen(f)+1=\Gen(g)+2\geq 2$, finishing the proof.
\end{proof}

\subsection{Competitiveness Analysis via Labeling Schemes}\label{app:generation-thm}

Using Lemma~\ref{lem:graphic-generation}, we show in \Cref{sec:simple-graphs} that Algorithm~\ref{alg:generation} is $(\frac{1}{4}p(1-p^2) - \frac{1}{2}p\ln(p))$-probability-competitive for simple graphs. in particular, we get a $0.2693$-probability-competitiveness when $p\approx 0.4485$.

Algorithm~\ref{alg:generation} may not achieve a sufficiently high competitiveness on general graphs with parallel edges, since the probability that a given edge $e^*\in \OPT(E)$ has generation $1$ may be much larger in this case. 
In \Cref{sec:general-graphs} we remedy this issue by considering an algorithm suitable for those elements and we combine it with Algorithm~\ref{alg:generation}, leading to an algorithm called \textsc{Mixture}, which is $0.2504$-probability-competitive.

\subsubsection{Simple Graphs: Proof of Theorem \ref{thm:graphic-intro}\ref{thm:graphic-simple}}\label{sec:simple-graphs}
We start by proving our competitiveness result for simple graphs.

\begin{theorem}\label{thm:generation}
Algorithm~\ref{alg:generation} is $(\frac{1}{4}p(1-p^2) - \frac{1}{2}p\ln(p))$-probability-competitive for simple graphs. 
\end{theorem}
In particular, Theorem \ref{thm:graphic-intro}\ref{thm:graphic-simple} follows as a direct corollary of Theorem \ref{thm:generation} since \Cref{alg:generation} is $0.2693$-probability-competitive when $p\approx 0.4485$. 
\begin{proof}[Proof of Theorem~\ref{thm:generation}]
Let $e^*\in \OPT(E)$ and $z$ be the improving word produced by the associated labeling scheme $\Lambda_1(e^*)$
in the interval $[p,1]$. For any subset of labels $S$, let $z_S$ be the restriction of $z$ to the labels in $S$ and set $\lambda_i=  i\ln(1/p)$ for all $i\in \NN$. By Lemma~\ref{lem:graphic-generation},
\begin{align*}
&\Pr[e^*\in \ALG_1]\\
&\geq \Pr[z_{\{1,2,3\}}=x1y\ \text{such that $x,y$ contains no $1$, $y$ does not end with $2$}]\\
&=\Pr[z_{\{1,2,3\}}=x1\ \text{such that $x\in \{2,3\}^*$}]+ \Pr[z_{\{1,2,3\}}=x1y'3\ \text{such that $x, y'\in \{2,3\}^*$}]. 
\end{align*}
We compute both terms separately. First, note that
\begin{align*}
&\Pr[z_{\{1,2,3\}}=x1\ \text{such that $x\in \{2,3\}^*$}]\\
&\quad= \Pr[z_{\{1,2,3\}}=x2\ \text{such that $x\in \{2,3\}^*$}]\\
&\quad= \frac{1}{2}\Pr[z_{\{1,2,3\}}=xc\ \text{such that $x\in \{2,3\}^*$, $c\in \{2,3\}$}] \\
&\quad= \frac{1}{2}\Pr[|z_{\{1\}}|=0]\cdot \Pr[|z_{\{2,3\}}|\geq 1] = \frac{1}{2} e^{-\lambda_1} (1-e^{-\lambda_2})= \frac{1}{2}p(1-p^2),
\end{align*}
where the first and second equalities follow since the symbols of $z$ are uniform in $\{1,2,3\}$, the third equality follows since $z_{\{1\}}$ and $z_{\{2,3\}}$ are independent random variables, and the last equality follows since the length of $z_S$ distributes as a Poisson with parameter $|S|\ln(1/p)$.
By a similar reasoning, the second term is
\begin{align*}
&\Pr[z_{\{1,2,3\}}=x1y'3\ \text{such that $x,y'\in \{2,3\}^*$}]\\
&\quad=
\frac{1}{2}\Pr[z_{\{1,2,3\}}=x1y'c\ \text{such that $x,y'\in \{2,3\}^*,c\in \{2,3\}$}] \\
&\quad= \frac{1}{2}\left(\Pr[z_{\{1,2,3\}}\ \text{has exactly one 1}] - \Pr[z_{\{1,2,3\}}=u1\  \text{such that $u\in \{2,3\}^*$}] \right) \\
&\quad= \frac{1}{2}\left(\Pr[|z_{\{1\}}|=1]  - \frac{1}{2}p(1-p^2)\right)\\
&\quad=\frac{1}{2}\left(\frac{e^{-\lambda_1}\lambda_1}{1!} - \frac{1}{2}p(1-p^2) \right)= \frac{1}{2}(p\ln(1/p)) - \frac{1}{4}p(1-p^2).
\end{align*}
Adding up the two terms, we conclude that $\Pr[e^*\in \ALG_1]\geq \frac{1}{4}p(1-p^2) - \frac{1}{2}p\ln (p).$
This expression attains its maximum at $p\approx 0.4485$, achieving a value of $0.2693$. 
\end{proof}

\subsubsection{General Graphs: Proof of Theorem \ref{thm:graphic-intro}\ref{thm:graphic-general}}\label{sec:general-graphs}
In what follows, we prove the $0.2504$-probability-competitiveness for general graphs. 
Consider the following algorithm:
\begin{algorithm}[H]
    \begin{algorithmic}[1]
    \State Initialize $\ALG_2 \gets \emptyset$.
    \State Skip all edges with arrival time in $[0,p)$.
    \State Let $X^p\subseteq E$ be the set of all parallel copies of the elements in $X=\OPT(E_p^+)$.
   \For{each edge $e$ with $t_e\in [p,1]$ in arrival order}
        \If{$e \in \OPT(E_{t_e}^+)\cap X^p$  and $\ALG_2 + e$ is independent}
                \State $\ALG_2 \gets \ALG_2 + e$
        \EndIf        
    \EndFor
    \State Return $\ALG_2$.
    \end{algorithmic}
    \caption{\textsc{Oblivious Algorithm}}
    \label{alg:parallel}
\end{algorithm}
We are ready to present our algorithm for general graphs. 
We combine \textsc{Generation} with \textsc{Oblivious}, leading to the following algorithm that we call \textsc{Mixture}.
\begin{algorithm}[H]
    \begin{algorithmic}[1]
    \State Fix the sample size $p \in (0,1)$.
    \State With probability $\varepsilon$, run \textsc{Oblivious} and return $\ALG_2$. 
    \State Otherwise, run \textsc{Generation} and return $\ALG_1$.
    \end{algorithmic}
    \caption{\textsc{Mixture Algorithm for Graphic Matroids}}
    \label{alg:final}
\end{algorithm}

As before, we work under Assumption~\ref{asm:graphic}, and we assume that $r\geq 3$.  For the analysis, we run the oblivious and the generation algorithm in parallel, simultaneously constructing $\AUX$, $\ALG_1$ (for \textsc{Generation}), and $\ALG_2$ (for \textsc{Oblivious}). Let $\ALG$ be the output of Algorithm~\ref{alg:final}.

Consider an element $e^*\in \OPT(E)$ with arriving time $t=t_{e^*}\in [p,1]$. For each $s\in [0,1]$, define $\Para(s)$ as the event that $\OPT(E_s^+ - \{e^*\})$ contains an edge  $f$ that is parallel to $e^*$ in the graphic matroid, and let $\alpha(s)=\Pr(\Para(s))$. Furthermore, for $0\leq a<b\leq 1$, define $\Para(a,b)$ as the event that there exists $s\in[a,b)$ for which $\Para(s)$ holds, 
and set $\alpha(a,b)=\Pr[\Para(a,b) | \neg \Para(b)]$.  We split the analysis into two cases, depending on whether $\Para(t)$ holds or not.\\

\noindent{\bf Case 1: $\Para(t)$ holds.}
In this case, there exists an element $f\in \OPT(E_t^+-\{e^*\})=\OPT(E_t)$ that is parallel to $e^*$. The arrival time $t_f$ of $f$ is uniformly distributed on $[0,t]$. If $t_f$ is in the interval $[0,p)$, then we know that $f\in \OPT(E_s)$ for all $s\in [p,t]$. In particular, $f\in \OPT(E_p^+)$ and $e^*$ is the first arriving element that is parallel to $f$ and improving, so it is selected by \textsc{Oblivious}.

Now focus on \textsc{Generation}. Let $t'$ be the largest improving time in $[0,t)$ such that the associated improving edge $e_{t'}$ is oriented toward some vertex in $\{u^*,v^*\}$. In terms of the labeling schemes $\Lambda_0(e^*)$ and $\Lambda_1(e^*)$ we have used so far for graphic matroids, this means that the label of $t'$ is in $\{1,2\}$. Note that if $t'<p$, then there are no arcs in $\AUX$ with head in $\{u^*,v^*\}$ when $e^*$ arrives. This means that $\Gen(e^*)=0$, therefore, $e^*$ is selected by the algorithm. 

Recall that $N_S[a,b)$ denotes the number of improving times in $[a,b)$ with label in $S$. Summarizing the above, by Lemma \ref{lem:Poisson-label}, we deduce that
\begin{align}
\Pr[e^*\in \ALG | \Para(t)]&\geq \varepsilon\Pr[t_f\in [0,p) \,|\, f\in \OPT(E_t)] + (1-\varepsilon)\Pr[t'< p | t^*=t]\nonumber\\
&=\varepsilon \Pr[N_{\{1\}}[p,t)=0]+ (1-\varepsilon) \Pr[N_{\{1,2\}}[p,t)=0]\nonumber\\
&= \varepsilon (p/t) + (1-\varepsilon) (p/t)^2. \label{eq:finalparallel}
\end{align}

\noindent{\bf Case 2: $\Para(t)$ does not hold.}
In this case, define $t'$ as before. 
Suppose first that $t'\geq p$. If $\Para(t',t)$ holds, then there exists a moment $s\in [t',t)$ in which $\Para(s)$ holds, i.e., $\OPT(E_s^+)$ contains an edge $f$ parallel to $e^*$. At that moment, $f$ is oriented toward $u^*$ or $v^*$, implying $t_f\leq t'$. We also conclude that no improving edge parallel to $f$ arrives in $(t_f,t)$. Furthermore, since  at time $t'$ exactly two edges of $\OPT(E_{t'}^+)$ are directed toward $u^*$ and $v^*$, we deduce that $f$ is the element arriving at time $t'$ with probability ${1}/{2}$, and $t_f<t'$ with probability ${1}/{2}$. Then, conditioned on $t_f\neq t'$, $t_f$ is a uniform random variable in $[0,t']$. Moreover, if $t_f\in [0,p)$ then, by the same argument as in Case 1, $e^*$ will be selected by \textsc{Oblivious}. By the above, if $p\leq s<t$, then $\Pr[e^*\in \ALG_2 | \neg \Para(t), t'=s, \Para(t',t)] = \frac{1}{2}(p/s)$. 

If $\Para(t',t)$ does not hold, then, since $\Para(t)$ does not hold by assumption, no improving element parallel to $e^*$ arrives in $[t',t]$. Let $g$ be the improving edge arriving at time $t'$ and $\vec{g}$ be its orientation. With probability ${1}/{2}$, the head of $\vec{g}$ is $u^*$ and with probability ${1}/{2}$ it is $v^*$. Suppose we are in the case $\vec{g}=(u', u^*)$. If $u'$ is the root $w$, then as long as no arc pointing toward $u^*$ or $v^*$ arrives in $[p,t')$, $\vec{g}$ and $\vec{e}$ will both enter $\AUX$, $\Gen(g)=1$ and $\Gen(e^*)=2$, so \textsc{Generation} selects $e^*$. If $u'$ is not the root $w$, then $u'\in V\setminus \{u^*,v^*\}$ as $g$ is not parallel to $e^*$. As long as no improving edge oriented toward $u^*$ or $v^*$ and at least one improving edge $h$ oriented toward $u'$ arrive in $[p,t')$, the arcs $\vec{h}$, $\vec{g}$ and $\vec{e}$ enter $\AUX$ and $\Gen(e)=\Gen(g)+1=\Gen(h)+2\geq 2$, so \textsc{Generation} selects $e^*$. By the above, $\Pr[e^*\in \ALG_1 | \neg \Para(t), t'=s, \neg \Para(t',t)] =  \frac{1}{2}\Pr[N_{\{1,2\}}[p,s)=0]\Pr[N_{\{3\}}[p,s)\geq 1] = \frac{1}{2}(p/s)^2(1-p/s)$. 

Let us now consider the case when $t'< p$. Note that in this case, we immediately have that $\Gen(e^*)=0$ and so $e^*$ is selected by \textsc{Generation}. For \textsc{Oblivious}, since no arc parallel to $e^*$ can possibly arrive in $(t',t)$, $e^*$ is selected by the algorithm as long as $\OPT(E_p^+)$ contains an element parallel to $e^*$, which is exactly event $\Para(p,t)$.

Putting everything together, using the fact that output of $\ALG$ is $\ALG_1$ with probability $(1-\varepsilon)$ and $\ALG_2$ with probability $\varepsilon$, as well as that the probability density function\footnote{ This comes from the fact that $t'$ has the same distribution as $S(t)$ for a rank-2 matroid. See Lemma \ref{lem:Poisson-all}. } of the variable  $t'$ is $f_{t'}(s)=2s/t^2$, we get that
\begin{align}
&\Pr[e^*\in \ALG | \neg \Para(t)]\nonumber\\ 
&\quad= \Pr[e^*\in \ALG | \neg \Para(t), t'<p] (p/t)^2+ \int_{p}^1 \Pr[e^*\in \ALG | \neg \Para(t), t'=s] f_{t'}(s) \diff s\nonumber\\
&\quad= \varepsilon\alpha(p,t)(p/t)^2+(1-\varepsilon)(p/t)^2\nonumber\\
&\qquad+ \int_{p}^t \left(\varepsilon\alpha(s,t)  \frac{1}{2}  (p/s)+(1-\varepsilon)(1-\alpha(s,t))  \frac{1}{2} (p/s)^2(1-p/s)\right) f_{t'}(s) \diff s\nonumber\\
&\quad= \varepsilon\alpha(p,t)(p/t)^2+(1-\varepsilon)(p/t)^2+ (p/t^2)\int_{p}^t \varepsilon\alpha(s,t)+(1-\varepsilon)(1-\alpha(s,t))(p/s)(1-p/s)\diff s.
\label{eq:finalnoparallel}
\end{align}
The bounds on the right-hand side of \eqref{eq:finalparallel} and \eqref{eq:finalnoparallel} are algebraic functions of $p$, $ \varepsilon$, and $t$. They also allude to probabilities  $\alpha(\cdot)$ and $\alpha(\cdot,\cdot)$ of random events, where the only stochastic part comes from the random arrival times of the elements -- they also depend on the adversarial choice of the graph and the value order on the edges.  We get a better understanding of \eqref{eq:finalparallel} and \eqref{eq:finalnoparallel} if we introduce an auxiliary random variable $M\in [0,1]$ that is defined as the maximum of all times $s<t$ such that $\Para(s)$ holds. Then, we can write $\alpha(t)=\Pr[M=t]$, $\alpha(s,t)\alpha(t)=\Pr[s\leq M<t]$ and $\alpha(s,t)(1-\alpha(t))=\Pr[M<s]$. In particular,  we can simplify  \eqref{eq:finalparallel} and \eqref{eq:finalnoparallel} using random variables as follows. For any fixed value $m^*\in [0,1]$, we have
\begin{align*}
&\Pr[e^*\in \ALG | t_{e^*}=t, M=m^*]\\
&\quad=\Pr[e^*\in \ALG | t_{e^*}=t, M=m^*, \Para(t)]\alpha(t)+ \Pr[e^*\in \ALG | t=t_e, M=m^*, \neg\Para(t)](1-\alpha(t)) \\
&\quad\geq \bm{1}[m^*=t]\cdot (\varepsilon (p/t)+(1-\varepsilon)(p/t)^2)+ \bm{1}[p\leq m^*<t]\cdot \varepsilon (p/t)^2+\bm{1}[m^*<t]\cdot (1-\varepsilon)(p/t)^2 \\
&\qquad  + (p/t^2) \int_{p}^t \left(\bm{1}[s\leq m^*<t]\cdot \varepsilon  + \bm{1}[m^*<s] \cdot (1-\varepsilon) (p/s)(1-p/s)\right) \diff s.
\end{align*}
Now, we define $G_1$, $G_2(m^*)$ and $G_3$ by computing the the expression above for specific values of $m^*$.
If $m^*=t$, then
\begin{align*}
\Pr[e^*\in \ALG \mid t_{e^*}=t, M=t] \geq \varepsilon (p/t)+(1-\varepsilon)(p/t)^2 \triangleq G_1.
\end{align*}
If $m^*\in [p,t)$, then
\begin{align*}
&\Pr[e^*\in \ALG \mid t_{e^*}=t, M=m^*]\\ 
&\quad\geq \varepsilon(p/t)^2 +(1-\varepsilon)(p/t)^2 +  (p/t^2)\int_{p}^{m^*}\varepsilon \diff s+ (p/t^2)\int_{m^*}^t(1-\varepsilon) (p/s)(1-p/s)\diff s\\
&\quad=(p/t)^2 + \frac{\varepsilon p(m^*-p)}{t^2} + (1-\varepsilon)(p/t^2)\left(p^2\left( 1/t-1/m^*\right) + p\ln(t/m^*)\right)\triangleq G_2(m^*).
\end{align*}
If $m^*<p$, then
\begin{align*}
&\Pr[e^*\in \ALG \mid t_{e^*}=t, M=m^*]\\
&\quad\geq (1-\varepsilon)(p/t)^2+ (p/t^2)\int_{p}^t(1-\varepsilon) (p/s)(1-p/s)\diff s\\
&\quad=(1-\varepsilon)(p/t)^2 + (1-\varepsilon)(p/t^2)\left(p^2(1/t-1/p) + p\ln(t/p)\right)\\
&\quad=(1-\varepsilon) \frac{p^2(p+t\ln(t/p))}{t^3}\triangleq  G_3.
\end{align*}
Observe that $G_3=G_2(p)-\varepsilon (p/t)^2$, and $G_1=G_2(t)$. Therefore,
\[
\min_{m^*\leq t} \Pr[e^*\in \ALG \mid t_{e^*}=t,M=m^*] = \min\set{G_3, \min_{m^*\in [p,t]} G_2(m^*)}.\]
It is not difficult to check that the second derivative of $G_2(m^*)$ with respect to $m^*$ is $(1-\varepsilon)(m^*-2p)p^2/((m^*)^3t^2)$. Therefore, for $m^*\leq 2p$, the function $G_2(m^*)$ is concave in $m^*$ and thus its minimum over the interval $[p,t]$ is attained in $m^*=p$ or in $m^*=t$. To ensure that $m^*\leq 2p$ holds, we assume from now on that $p\geq {1}/{2}$. 
 Since $G_2(t)=G_1$ and $G_2(p)\geq G_3$, we obtain
\begin{align}
\Pr[e^*\in \ALG \mid t_{e^*}=t] &\geq \min_{m^*\leq t} \Pr[e^*\in \ALG \mid t_{e^*}=t,M=m^*] = \min\{G_3, G_1\}. \label{eq:final}
\end{align}

Based on the previous analysis, we are now ready to lower bound the probability-competitiveness of \textsc{Mixture} to prove Theorem \ref{thm:graphic-intro}\ref{thm:graphic-general}.

\begin{proof}[Proof of Theorem~\ref{thm:graphic-intro}\ref{thm:graphic-general}]
We choose $e^*\in \OPT(E)$ with arriving time $t=t_{e^*}\in [p,1]$. Let $\ALG$ denote the output of Algorithm~\ref{alg:final}.
From \eqref{eq:final}, we get that
\begin{align*}
\Pr[e^*\in \ALG]&\geq \int_p^1 \min\left(\varepsilon (p/t) + (1-\varepsilon)(p/t)^2, (1-\varepsilon) (p/t)^3(1+(t/p)\ln(t/p))\right)\diff t\\
&=p\int_{p}^1\min\left(\varepsilon/q+1-\varepsilon,(1-\varepsilon)(q-\ln (q))\right)\diff q,
\end{align*}
where the second equality holds by performing a change of variables $q={p}/{t}$.
Fix a value of $p\in [{1}/{2},1]$ and consider the two functions involved in the minimum inside he integral, that is, $K_1(\varepsilon,q)=\varepsilon/q + 1-\varepsilon,\text{ and }K_2(\varepsilon,q)=(1-\varepsilon)(q-\ln(q)).$
Note that for every given value $q> p$, both $K_1(\varepsilon,q)$ and $K_2(\varepsilon,q)$ are linear functions in $\varepsilon$.
Furthermore, we have $K_1(1,q)={1}/{q}>K_2(1,q)=0$. Thus, $K_1(\varepsilon,q)\leq K_2(\varepsilon,q)$ if and only if 
\begin{equation}
\varepsilon\leq \frac{q-\ln(q)-1}{q-\ln(q)-1+1/q}=R(q).\label{ineq:eps-final}
\end{equation}
The function $R(q)$ is decreasing in the interval $[{1}/{2},1]$, with $R(1)=0$.
If $\varepsilon>R({1}/{2})$, inequality \eqref{ineq:eps-final} is not satisfied for any $q\in [{1}/{2},1]$, and therefore $K_1(\varepsilon,q)>K_2(\varepsilon,q)$ for every $q\in [p,1]$.
If $\varepsilon \in [0,R({1}/{2})]$, there exists a unique value $q_{\varepsilon}\in [{1}/{2},1]$ such that $\varepsilon=R(q_{\varepsilon})$.
Then, 
\[\min(K_1(\varepsilon,q),K_2(\varepsilon,q)) = 
\begin{cases} 
K_2(\varepsilon,q) &\text{ if $\varepsilon>R({1}/{2})$,}\\
K_2(\varepsilon,q) &\text{ if $\varepsilon\le R({1}/{2})$ and $q> q_\varepsilon$,}\\
K_1(\varepsilon,q) &\text{ if $\varepsilon\le R({1}/{2})$ and $q\le q_\varepsilon$.}
\end{cases}\]
In particular, the lower bound for $\Pr[e^*\in\ALG]$ can be rewritten as
\begin{equation*}
p\int_p^1(1-\varepsilon)(q-\ln(q))\diff q=p(1-\varepsilon)(-p^2/2 - p + p\ln(p) + 3/2)
\end{equation*}
if $\varepsilon>R({1}/{2})$,
\begin{equation*}
p\int_p^1(1-\varepsilon)(q-\ln(q))\diff q=p(1-\varepsilon)(-p^2/2 - p + p\ln(p) + 3/2)
\end{equation*}
if $\varepsilon\le R({1}/{2})$ and $p\ge q_{\varepsilon}$, and
\begin{align*}
&p\int_{p}^{q_{\varepsilon}}(\varepsilon/q+1-\varepsilon)+p\int_{q_{\varepsilon}}^1(1-\varepsilon)(q-\ln(q))\diff q\\
&\quad=p\varepsilon\ln(q_{\varepsilon}/p)+(1-\varepsilon)(q_{\varepsilon}-p)p+p(1-\varepsilon)(-q_{\varepsilon}^2/2 - q_{\varepsilon} + q_{\varepsilon}\ln(q_{\varepsilon}) + 3/2) 
\end{align*}
if $\varepsilon\le R({1}/{2})$ and $p\le q_{\varepsilon}$.

To conclude, we solve an optimization problem in each of the three regions defined by these cases to find the optimal combination for $\varepsilon$ and $p$. For the case $\varepsilon>R({1}/{2})\approx 0.0881$, the maximum of this bound is attained when $p={1}/{2}$, yielding a maximum value of $\approx 0.2409$. When $\varepsilon\le R({1}/{2})$, to find the maximum, we set $\varepsilon=R(q)$ and obtain an explicit maximization problem in two variables with constraints ${1}/{2}\leq q\leq p\leq 1$ for the second region, and ${1}/{2}\leq p\leq q\leq 1$ for the third region, and the maximum can be computed numerically in both cases. In the former, we get an optimal value of $0.2409$. However, for the latter, the optimal solution is given by $p^{\star}=0.5$ and $q^{\star}\approx 0.8251$, with $\varepsilon=R(q^{\star})\approx 0.0141$, showing a probability-competitive ratio of at least $0.2504$.
\end{proof}
%%%%%%%%%%%%%%%%
\noindent{\bf Acknowledgements.} This research was supported by the Lend\"ulet Programme of the Hungarian Academy of Sciences -- grant number LP2021-1/2021, by the Ministry of Innovation and Technology of Hungary from the National Research, Development and Innovation Fund -- grant numbers ADVANCED 150556 and ELTE TKP 2021-NKTA-62, by Dynasnet European Research Council Synergy project -- grant number ERC-2018-SYG 810115, and by ANID-Chile -- grants FONDECYT 1231669, FONDECYT 1241846, BASAL Center for Mathematical Modeling FB210005, and ANILLO Information and Computation in Market Design ACT210005.

\bibliographystyle{abbrv}
\bibliography{references}

\appendix
\section{Forbidden Sets Technique through Labeling Schemes}
\label{app:qforbidden}
%%%%%%%%%%%%%%%%

We briefly describe the $q$-forbidden technique of Soto, Turkieltaub, and Verdugo~\cite{forbidden-paper} and show how to implement every $q$-forbidden algorithm via a labeling scheme. 
The $q$-forbidden algorithmic scheme of \cite{forbidden-paper} also starts by first sampling the elements arriving up to a given time $p$ and then accepting a subset of the improving elements. Specifically, for every $e^* \in \OPT$ arriving at a given time $t$, the algorithm identifies, for every time $s < t$, a set $\cF(\OPT(E_s),\OPT(E_t), e^*)$ of at most $q$ \emph{forbidden} elements such that if no improving element arriving at $s \in (p,t)$ is in $\mathcal{F}(\OPT(E_s),\OPT(E_t), e^*)$ then $e^*$ is added to the output set $\ALG$.

\begin{proposition}
For every matroid class $\cC$ that admits a $q$-forbidden algorithm $\ALG$, there exists a labeling scheme $\Lambda$ such that, for every $p \in (0,1)$ and $e^*\in \OPT$,
\[
\Pr[e^* \in \ALG] = \begin{cases}
-p \ln(p)           & \text{for $q = 1$,} \\
{(p-p^q)}/{(q-1)}   & \text{for $q \geq 2$.}
\end{cases}.
\]
Choosing the optimal sample size $p = p(q)$ yields a $c(q)$-probability-competitive ratio, where $(p(1), c(1)) = (1/e, 1/e)$ and $(p(q), c(q)) = \prn{q^{-1/(q+1)}, q^{-q/(q+1)}}$.
\end{proposition}
\begin{proof}
Let $e^*\in\OPT(E)$. Our labeling scheme is dynamic. We start by considering a labeling scheme $\Lambda_\pi$ induced by a total order $\pi$ of all elements of $E$ for which $\pi_1 = e^*$, and use it for the interval $[t_{e^*}, 1]$. For times $s < t_{e^*}$, we use a different labeling. By the definition of a $q$-forbidden algorithm, $e^*$ can be selected if no element in $\cF(E_s, E_{t_{e^*}},e^*) \subseteq \OPT(E_s)$ arrives in $[p,s)$, where $\abs{\cF(E_s, E_{t_{e^*}},e^*)}\leq q$. Our labeling scheme gives arbitrary labels from $[q]$ to the elements in $\cF(E_s, E_{t_{e^*}},e^*)$, and arbitrary labels from $[r] \setminus [q]$ to the remaining elements of $\OPT(E_s)$. In particular, any element with label greater than $q$ is not part of $\cF(E_s, E_{t_{e^*}}, e^*)$. Let us denote the labeling scheme thus obtained by $\Lambda$.

Consider the improving word $z$ in $[p, 1]$ given by $\Lambda$. By Lemma~\ref{lem:standardlabel}, the first symbol $1$ in $z$ represents the arrival of $e^*$. From that point on, going backwards in time, labels in $\{1, 2, \dots, q\}$ correspond to forbidden elements for $e^*$. Thus,
\begin{align*}
\Pr[e^* \in \ALG] \geq \Pr[z = a1b \text{ such that}\ a\in \prn{[r] \setminus 1}^*, b \in \prn{[r] \setminus [q]}^*].
\end{align*}
Instead of using the word $z$, we use the simpler restriction $z'$ of $z$ to the alphabet $[q]$. By Lemma~\ref{lem:Poisson-label}, $z'$ is a uniform random word in ${[q]}^*$ whose length distributes as a Poisson distribution with parameter $\lambda_q = q \ln\prn{{1}/{p}}$. We separate the analysis for $q = 1$ due to its simplicity, giving
\begin{equation}\label{eq:forbidden-1}
\Pr[e^* \in \ALG] \geq \Pr[z' = 1] = \frac{\lambda_1^1 e^{-\lambda_1}}{1!} = - p\ln(p).
\end{equation}
For $q \geq 2$, let $\lambda' = \lambda_q \cdot \prn{1 - 1/q}= (q-1) \ln\prn{1/p}$. Then,
\begin{align}
\Pr[e^*\in \ALG] &\geq \Pr\brk{z' = a'1 \midd a \in \prn{[q] \setminus \{1\}}^*} \nonumber \\
&= \sum_{k = 1}^{\infty} \frac{\lambda_q^k e^{-\lambda_q}}{k!} \prn{1 - \frac{1}{q}}^{k-1}  \frac{1}{q}\nonumber\\
&=\sum_{k = 1}^{\infty} \frac{\prn{\lambda'}^k e^{-\lambda'}}{k!} \frac{1}{q - 1} e^{\lambda' - \lambda_q} \nonumber \\
\label{eq:forbidden-2} &= \frac{1 - e^{-\lambda'}}{q-1} e^{\lambda'-\lambda_q} = \frac{1-p^{q-1}}{q-1} e^{-\ln\prn{1/p}} = \frac{p-p^q}{q-1}.
\end{align}
Optimizing~\eqref{eq:forbidden-1} and~\eqref{eq:forbidden-2}, we get that the optimal $p$'s are $p(1) = 1/e$ and $p(q) = q^{-1/(q+1)}$ for $q \geq 2$, yielding a probability-competitive ratio of $c(1) = 1/e$ and $c(q) = q^{-q/(q+1)}$ for $q \geq 2$.
\end{proof}
%We remark that the previous result gives an alternative proof of, for example, the $1/e$-competitive algorithm for transversal matroids \cite{kesselheim-transversal-mat-sec} (which are 1-forbidden).
%%%%%%%%%%%%%%%%
\section{Competitiveness for Small Ranks}\label{app:tables}

Proposition~\ref{lem:gi} gives an exact formula that can be seen as a polynomial in $p$ and $\ln(p)$ of total degree at most $2r$. Hence, it is easy to evaluate for small values of $r$. Table~\ref{tab:GreedyUniform} presents the values of $c(r,p)$, $p(r)=\arg\max\{c(r,p)\colon p\in (0,1)\}$ and $c(r,p(r))$ for small values of $r$.

\begin{table}[h]
\centering
\begin{tabular}{llll}
\toprule
\textbf{$r$} & \textbf{$c(r,p)$} & \textbf{$p(r)$} & \textbf{$c(r,p(r))$} \\
\midrule
1 & $-p\ln(p)$ & $1/e \approx 0.3678$ & $1/e \approx 0.3678$ \\
2 &$p  (3 - 3 p + 2 p \ln(p)) $& $0.3824$ &$0.4273$  \\
3 &$(p/8)  (19 - 19  p^2 + 30  p^2  \ln(p) - 18  p^2  \ln(p)^2)$& $0.3867$ & $0.4575$ \\
4 &$(p/81)  (175 - 175  p^3 + 444  p^3  \ln(p) - 
   504  p^3  \ln(p)^2 + 288  p^3  \ln(p)^3)$ & $0.3883$ & $0.4769$ \\
\bottomrule
\end{tabular}
\caption{\textsc{Greedy-Improving} is $c(r,p(r))$-probability-competitive for uniform matroids of rank $r$.}
\label{tab:GreedyUniform}
\end{table}

Next, we consider general laminar matroids. Using that  $\Pr[P(\lambda')\geq r]=1-\Pr[P(\lambda')<r]$ and the exact expression for a Poisson distribution taking a value of $k$ for $k\in \{0,\dots, r\}$, we get a formula for $a(r,p)$ that can be seen as a polynomial in $p$ and $\ln(p)$ of total degree at most $2r$, with coefficients that depend on $r$. Hence, it is easy to evaluate for small values of $r$. Table~\ref{tab:table-laminar} presents the values of $a(r,p)$, $p(r)=\arg\max\{a(r,p)\colon p\in (0,1)\}$ and $a(r,p(r))$ for small values of $r$.

\begin{table}[h!]
\centering
\begin{tabular}{llll}
\toprule
\textbf{$r$} & \textbf{$a(r,p)$} & \textbf{$p(r)$} & \textbf{$a(r,p(r))$} \\
\midrule
1 & $-p\ln(p)$ & $1/e \approx 0.3678$ & $1/e \approx 0.3678$ \\
2 &$p(2 - 2 p + p \ln(p))$& $0.4241$ &$0.3341$  \\
3 &$(p/8) (11 - 11 p^2 + 14 p^2 \ln(p) - 6 p^2 \ln(p)^2)$& $0.4490$ & $0.3225$ \\
4 &$(p/81) (94 - 94 p^3 + 201 p^3 \ln(p) - 180 p^3 \ln(p)^2 + 72 p^3 \ln(p)^3)$& $0.4629$ & $0.3169$\\
\bottomrule
\end{tabular}
\caption{\textsc{Greedy-Improving} is $a(r,p(r))$-probability-competitive for laminar matroids of rank $r$.}
\label{tab:table-laminar}
\end{table}
\end{document}